\documentclass[11pt]{article}
\usepackage{hyperref}

\usepackage[margin=1in]{geometry}
\usepackage{amsmath,amssymb}
\usepackage{mathtools}

\usepackage{amsthm}
\usepackage{graphicx}
\usepackage[x11names]{xcolor}
\usepackage[textsize=tiny, backgroundcolor=Khaki1, linecolor=black, textwidth=1.5in, disable]{todonotes}
\usepackage{setspace}
\usepackage{placeins}

\newtheorem{proposition}{Proposition}
\newtheorem*{splitting-stability}{Assumption (Splitting Stability)}
\newtheorem{assumption}{Assumption}
\theoremstyle{remark}
\newtheorem*{remark}{Remark}

\DeclareMathOperator{\E}{\mathbf{E}}
\let\P\relax\DeclareMathOperator{\P}{\mathbf{P}}
\DeclareMathOperator{\Var}{\mathbf{Var}}
\DeclareMathOperator{\Cov}{\mathbf{Cov}}
\DeclareMathOperator{\tr}{tr}
\DeclareMathOperator*{\argmax}{arg\,max}
\DeclareMathOperator*{\argmin}{arg\,min}
\DeclareMathOperator{\TV}{TV}
\DeclareMathOperator{\dist}{dist}
\DeclareMathOperator{\1}{\mathbf{1}}

\def\R{\mathbf{R}}
\def\X{\mathcal{X}}

\DeclareMathOperator{\RF}{RF}

\def\Dto{\xRightarrow{\;\rm dist\;}}
\def\Pto{\xrightarrow{\;\P\;}}
\def\coloneqq{\coloneq}

\usepackage{microtype}

\title{Asymptotic Normality for Multivariate Random Forest Estimators}
\author{Kevin Li\\ Massachusetts Institute of Technology \\ \href{mailto:kkli@mit.edu}{kkli@mit.edu}
\thanks{I would like to thank my advisors Alberto Abadie and Victor Chernozhukov for reviewing multiple drafts of this paper.
In addition, Sophie (Liyang) Sun, Ben Deaner, participants in the seminar class 14.386 (Spring 2020), and participants
of MIT's Econometrics Lunch seminar provided very useful feedback. I would also like to thank Professor Stefan Wager
for helping me understand mechanics of covariance estimates in random forest models.}}

\date{December 16, 2020}
\begin{document}
\maketitle
\begin{abstract}
\noindent
  Regression trees and random forests are non-parametric estimators
that are widely used in data analysis. A recent paper by
Athey and Wager shows that the pointwise random forest estimate
is asymptotically Gaussian. In this paper, we extend their result to the
multivariate case and show that a vector of estimates, taken at multiple
points, is jointly asymptotically normal. Specifically, the covariance matrix of the
limiting normal distribution is diagonal, so that the estimates at any
two points are independent in sufficiently deep trees. We show that
the off-diagonal terms are bounded by quantities related to the probability
of two given points belonging in the same leaf of the resulting tree. Our results
rely on certain stability properties of the underlying tree estimator,
and we give examples of splitting rules for which this holds.
We also provide a heuristic and numerical simulations
for gauging the decay of the off-diagonal term in finite samples.
\end{abstract}
\onehalfspacing
\section{Introduction}
Trees and random forests are non-parametric estimators first
introduced by Breiman \cite{breiman01}.  Given a feature space
$\X \subset \R^p$ and a set of data points
$\{ (X_i, Y_i) \} \subseteq \X \times \R$, tree estimators recursively
partitions the feature space into axis-aligned non-overlapping
hyperrectangles\footnote{When the feature space $\X$ need not be
  rectangular, one may always enlarge $\X$ to a rectangular set $\X'$
  that is defined to the intersection of all rectangular sets
  containing $\X$.} by repeatedly splitting $\X$ along a given axis.
The prediction of the tree estimator at a test point $x \in \X$ is
then an aggregate of the targets $Y_i$'s that land in hyperrectangle
containing $x$; when $Y_i$ is continuous, the aggregate is the sample
mean and the tree is also known as a regression tree. The depth of a
tree estimator---defined as the maximal number of splits taken before
reaching a terminal hyperrectangle---controls the complexity of the
tree estimator. There are two popular methods for controlling
complexity: the ``boosting'' approach grows trees of large depth, then
reduces complexity by either trimming the tree (i.e., so that
predictions are made at a non-terminal hyperrectangle) or introducing
a decay factor; the ``bagging'' approach instead grows a collection of
shallow trees on different subsets of the data, and averages over the
trees for the final prediction.  The intuition for bagging is that
trees grown on different subsets are not perfectly correlated, so that
aggregation reduces variance and balances the bias-variance
tradeoff. Estimators of this type are called random forests, and they
are the focus of this paper.

Since their introduction in the early 2000s, random forests have become an increasingly important
tool in applied data analysis, owing to a multiple of practical advantages over competing models.
First, high-quality random forest libraries are readily available, with popular implementations
that scale to hundreds of distributed workers \cite{lightgbm,xgboost}. Moreover, the core algorithm
behind tree estimators and random forests are simple enough to allow for rapid prototyping
of bespoke implementations, e.g. \cite{grfgithub}. Another advantage of tree-based methods
is that they can ingest real-world data without much issue: continuous, discrete, and ordered
categorical features
may be freely mixed\footnote{Splits on discrete features partitions that variable
into two arbitrary non-empty sets; no changes are needed for ordered categorical features.}~\cite{catboost},
model estimates are immune to feature outliers, and missing data may be easily incorporated.
Firstly, their construction
naturally aligns with the spatial locality of most applied: that is, the underlying target function relating
$Y$ to $X$ is continuous. Finally, tree models are interpretable, with well-defined
notions of feature importance \cite{gregorutti2017correlation,strobl2008conditional}, which
supports their use as model selection tools \cite{genuer2010variable}.

Within economics, random forests may be fruitfully applied to estimate heterogeneneous
treatment effects. In Rubin's potential outcomes framework \cite{rubin} (see \cite{imbens_rubin_2015}
for an overview), an individual $i$ is associated with two potential outcomes $Y^{(0)}$
and $Y^{(1)}$, with one of the outcomes being realized depending on whether $i$ undergoes treatment.
The statistician has access the IID observations $\{ (X_i, W_i, Y_i) : 1 \leq i \leq n \}$, where
$X_i$ is a vector of observed covariates for individual $i$, $W_i \in \{ 0, 1 \}$ is (an encoding of)
their treatment status, and $Y_i = Y_i^{(W_i)}$ is her realized outcome. One quantity
of interest is the treatment effect at $x$
\begin{equation}
\tau(x) \coloneq \E(Y^{(1)}  - Y^{(0)} \mid X_i = x).
\end{equation}
Since only one of $Y^{(0)}_i$ and $Y^{(1)}_i$ is observed, consistent estimation of $\tau(x)$
requires further distributional assumptions. A common assumption is unconfoundedness, i.e.,
that treatment status $W_i$ is independent of $Y^{(1)}$ and $Y^{(0)}$ conditional on $X_i$. Under this assumption,
\begin{equation}
\tau(x) = \E \biggl[
Y_i \biggl( \frac{W_i}{e(x)} - \frac{1-W_i}{1-e(x)} \biggr) \mid X_i = x
 \biggr], \quad \text{where $e(x) = \P(W_i = 1 \mid X_i = x)$.}
\end{equation}
Here, the key function is $e(x)$, known as the propensity score, is the probability of treatment for
the subpopulation with covariates $x$; see \cite{propensity_score} the derivation and implications.
Machine learning methods---including random forests---may be brought to bear on the
problem by estimating $e(x)$.
Alternatively, unconfoundedness also implies
\begin{equation}
\label{eq:21}
\tau(x) = \E(Y \mid W = 1, X = x) - \E(Y \mid W = 0, X = x),
\end{equation}
so that $\tau(x)$ may be estimated by fitting two models, one
on the subset of the sample in which $W = 1$, and the other on $W = 0$.

In econometric applications, conducting pointwise inference on
the target function $f: \X \to \R$ (e.g., to test the null hypothesis $H_0: f(x) = 0$)
requires knowledge about
about the rate of convergence or asymptotic distribution of the underlying
estimator $\hat f(x)$, where $x$ is the point of interest. However,
functionals of target function are often also of interest:
for example, the difference of treatments effects (i.e., $f = \tau$) for two different subpopulations
is captured by the quantity
\begin{equation}
f(x) - f(\bar x),
\end{equation}
where $x$ and $\bar x$ are covariates describing the two subpopulations.
More generally, we might also be interested in a weighed treatment effect, where
a subpopulation $x$ is given an importance weight modeled as a density $\mu(x)$.
In this case, the corresponding functional of $f$ is
\begin{equation}
\int_{x \in \mathcal{X}} f(x) d\mu, \quad \text{where $\mu$ is not necessarily the density of $x$},
\end{equation}
and the integral is taken over the domain $\mathcal{X}$.

Inference on functionals of $f$ requires not only the asymptotic distribution of the point
estimate $f(x)$, but also the correlation between estimates at different
points $f(x)$ and $f(\bar x)$. As a concrete example, consider the function $\tau(x)$
and the simple difference $\tau(x) - \tau(\bar x)$. We have
\begin{equation}
\label{eq:23}
\begin{split}
\tau(x) - \tau(\bar x) &=
[\E(Y \mid W = 1, X = x) - \E(Y \mid W = 1, X = \bar x)] \\
&\qquad \qquad-
[\E(Y \mid W = 0, X = x) - \E(Y \mid W = 0, X = \bar x)]\\
&\eqqcolon A - B.
\end{split}
\end{equation}
We may estimate the difference by estimating $A$ and $B$ separately,
fitting a random forest model to the two ``halves'' of the dataset where $W_i = 1$
and $W_i = 0$, as discussed above. The estimators $\hat A$ and $\hat B$ obtained
are thus independent, so that  $\Var(\hat A - \hat B)=\Var \hat A + \Var \hat B$. 
The variances $\hat A$ and $\hat B$ then depend on the covariance of their respective
random forest estimates at $x$ and $\bar x$.

This paper studies the correlation structure of a class
of random forests models whose asymptotic distributions were first worked out in \cite{crf}.
We find sufficient conditions under which the asymptotic covariance of random forest
estimates at different points vanish relative to their respective variances;
moreover, we provide finite sample heuristics
based on our calculations. To the best of our knowledge, this is the first set of results
on the correlation structure of random forest estimators.

The present paper builds on and extends the results in \cite{crf}, which in turn
builds on related work \cite{Wager2015AdaptiveCO}
on general concentration properties of trees and random forest estimators.
See also \cite{athey2019}, which extends the random forest model considered here to
a broader class of target functions by incorporating knowledge of
moment conditions. Stability results established in this paper have appeared
in \cite{arsov2019stability}, who study the notions of algorithmic stability for
random forests and logistic regression and derive generalization error guarantees. Also closely related
to our paper are \cite{chernozhukov2017} and \cite{chen2018},
concerning finite sample Gaussian approximations of sums and $U$-statistics
in high dimensions. In this context, our paper provides a stepping stone towards applying the theory
of finite sample $U$-statistics to random forests, where bounds on covariance matrices
plays a central role.

The paper is structured as follows. In Section 2, we introduce the random forest model
and state the assumptions required for our results; Section 3 contains our main theoretical
contributions; Section 4 builds on Sections 3 and discusses heuristics useful in finite sample settings;
Section 5 concludes. All proofs are found in the appendix.

\section{Model Setup and Assumptions}
\subsection{Overview of Tree Estimators}
The goal of this paper is to study asymptotic Gaussian approximations
of random forest estimators.  Throughout, we assume that a random
sample
$\{ Z_i = (X_i, Y_i) : 1 \leq i \leq n \} \subseteq \X \times \R$ is
given, where each $X_i$ is a vector of \emph{features} or
\emph{covariates} belonging to a subset $\X \subseteq \R^p$ of
$p$-dimensional Euclidean space, and $Y_i \in \R$ is the \emph{response} or \emph{target}
corresponding to $X_i$.  We will refer to $\X$ as the
feature space or the feature domain.

Given the data set $\{ Z_i \}_{i=1}^n$, a tree estimator recursively partitions
the feature space by making axis aligned splits. Specifically, an axis-aligned split
is a pair $(j, t)$ where $j \in \{ 1, \dots, p \}$ is the \emph{splitting coordinate}
and $t \in \R$ is the \emph{splitting index}; given a subset $\mathcal{R} \subseteq \X$,
a split $(j, t)$ divides $\mathcal{R}$ into left and right halves
\begin{equation}
\label{leftright}
\{ x \in \mathcal{R} : x_j < t \}
\quad \text{and} \quad
\{ x \in \mathcal{R}: x_j > t \},
\end{equation}
where $x_j$ denotes the $j$-th coordinate of the vector $x$. Starting with the
entire feature space $\X$, the recursive splitting algorithm
computes a (axis-aligned) split
based on the data $\{ Z_i: 1 \leq i \leq n \}$; for example, when the target $Y_i$
continuous, a popular choice is
\begin{equation}\label{square-split}
(j, t) = \argmin_{\tilde j, \tilde t}
\sum_{i : X_i \in L} (Y_i - \mu_L)^2  + \sum_{i : X_i \in R} (Y_i - \mu_R)^2
\end{equation}
where $L = L(\tilde j, \tilde t)$ and $R = R(\tilde j, \tilde t)$ are
the two halves of $\X$ obtained by the split $(\tilde j, \tilde t)$,
with $\mu_L$ and~$\mu_R$ being the averages of targets $Y_i$ whose
corresponding feature land in $L$ and $R$, respectively.

After the first split, $\X$ is split into two halves $L$ and $R$.
The process is then repeated for $L$ and $R$ separately, in that a split for $L$
is computed by using the subset of the data whose features $X_i$ belong in $L$,
and likewise for $R$. Each of the halves is then split again, and so on,
until a stopping criterion is met. The process completes
when the stopping criterion is satisfied for each node; at this point,
the collection of halves form a
partition\footnote{According to ~\eqref{leftright}, we exclude edge
cases where points lead on the an ``edge'' of the rectangle. This is not an issue
for continuous variables, while for categorical variables, the definition
should be slightly changed so that one of the halves contain $x_j = t$.
In this paper, we deal only with continuous features.} of $\X$, with each partition---and
all the halves that came before it---being the
intersection of a hyperrectangle with $\X$.
The sequence of splits corresponds to a tree in the natural way; we will
call the halfspaces that arise during the splitting
as \emph{nodes}, and elements of the final
partition \emph{terminal nodes}.

Given the collection $N_1, \dots, N_q$ of terminal nodes (which form a partition
of $\X$), the prediction of the tree at generic test point $x \in \X$
is the average of the responses that belong in the same terminal node as $x$
\begin{equation}\label{final-prediction}
T(x; \xi, Z_1, \dots, Z_n) = \sum_{j=1}^q \1(x \in N_j)
\frac{1}{|N_j|} \sum_{i : X_i \in N_j} Y_j,
\end{equation}
where the outer sum runs over observations $i$ for which $X_i$
belongs to the partition $N_j$, and $|N_j|$ is the number of such observations.
The input $\xi$ is an external source of randomization to allow for randomized
split selection procedures. Thus, $T(x; \xi, Z_1,\dots,Z_n)$
refers to prediction at $x$ for a tree grown using data $\{ Z_1, \dots, Z_n \}$
with randomization parameter $\xi$. As a function of $x$, keeping $Z_1, \dots, Z_n$
and $\xi$ fixed, $T:\X \to \R$ is then a step function, i.e., a linear combination
of indicator functions of rectangular sets.

We note here that equations \eqref{square-split} and \eqref{final-prediction}
are not the only possible choices; in particular, the rule used to choose
the optimal split may be path dependent (i.e., dependent on previous splits)
as in popular implementations (see \cite{xgboost}, which allows
for random forests but uses gradient boosting after the initial split), and
the final prediction rule \eqref{final-prediction} may instead do a final linear fit
or use weighted average (c.f., \cite{lightgbm,xgboost}). Analysis of
tree (and random forest) models is complicated in these cases, so the present paper
stipulates that the algorithm uses a splitting rule that is ``similar'' to
similar to \eqref{square-split} (c.f. Proposition~\ref{prop:split-stability-sufficient})
and uses \eqref{final-prediction} as the final prediction. In particular,
this implies that the target estimated by the tree estimator is the regression
function $x \mapsto \E(Y \mid X = x)$.

\subsection{From Trees to Random Forests}
Given a specific tree estimator $T$, i.e., given a set of splitting rules
which defines the tree estimator, we define the random forest estimator
to be the average of the tree estimator
across all $\binom{n}{s}$ subsamples, marginalizing over the randomization
device $\xi$. Specifically, the random forest estimate $\RF(x)$ at $x \in \X$,
given data $\{ Z_1, \dots, Z_n \}$, is defined to be
\begin{equation}
\label{rf-definition}
\RF(x; Z_1, \dots, Z_n)
=
\frac{1}{\binom{n}{s}}
\sum_{i_1, \dots, i_s}
\E_{\xi} T(x; \xi, Z_{i_1}, \dots, Z_{i_s}),
\end{equation}
where the summation runs over size-$s$ subsets of
$\{ 1, \dots, n \}$, and the inner expectation
is taken with respect to $\xi$. Importantly,
note each tree is grown on a subsample of size $s<n$. We
follow \cite{crf} in assuming that $s \sim n^{\beta}$
for some $\beta$ sufficiently close to one; specifically,
we assume throughout that the subsample size is chosen
as to satisfy the assumptions of Theorem~3
of \cite{crf}, so that---along with other assumptions to be introduced
presently---the random forest estimator $\RF$ is a consistent
estimator of the target function $x \mapsto \E(Y \mid X = x)$
(see discussion above).

In keeping with the notation
\cite{crf}, we will write $T(x; Z_1, \dots, Z_s)$
to refer to the expectation of $T(x; \xi, Z_1, \dots, Z_s)$ over $\xi$.
With this notation, the random forest (at $x$)
estimator \eqref{rf-definition} is a $U$-statistic
with the size~$s$ kernel $(Z_1, \dots, Z_s) \mapsto T(x, Z_1, \dots, Z_s)$.
We will discuss the $U$-statistic representation of $\RF$ in greater
detail in Section 3.

\subsection{Discussion of Model Assumptions}
As our results will be an extension of the results in \cite{crf},
we will study the same model of random forests and adopt a similar set of assumptions.
The assumptions regarding tree estimators
have appeared before in \cite{Wager2015AdaptiveCO}, while the
distributional assumptions on the conditional moments of $Y$ are standard
(see e.g., Chapters 7 and 9 in \cite{hastie2009elements}).

The first---and most bespoke---assumption is that the tree algorithm
is honest. Intuitively, honesty stipulates that that knowledge of the tree structure
does not affect the conditional distribution of tree estimates when
the features are fixed.
\begin{assumption}[Honesty]
The target $Y_i$ and the tree structure (i.e., the splitting coordinates and splitting indices) 
are independent conditional on $X_i$. Specifically, we require
\begin{equation}\label{honesty-assumption}
\dist(Y_i \mid X_i, S)  = \dist(Y_i \mid X_i),
\end{equation}
for all observations $i$ where $Y_i$ participates in the final prediction,
where $S$ is set of splits chosen by the tree algorithm. (The second
equality is automatic as a consequence of independent observations.)
\end{assumption}
There are several ways to satisfy this assumption. The first is to calculate splits
based only the features $X_i$ only. This rules out out the example
splitting rule given in \eqref{square-split}, so we could instead
use its analog in feature space
\begin{equation}
(j, t) = \argmin_{\tilde j, \tilde t}
\sum_{i : X_i \in L} \| X_i - \mu_L \|_2^2  + \sum_{i : X_i \in R} \| X_i - \mu_R \|^2,
\end{equation}
where here $\mu_L$ and $\mu_R$ denote the average (i.e., center of mass)
of the $X_i$ in each halfspace. In this instance, the choice
of splits is essentially a clustering algorithm that finds the best
division of the sample points into two parts.
Another way to satisfy to the honesty assumption while still computing
splits based on the targets is to use sample splitting.  The data is
partitioned into two parts $\mathcal{I}_1$ and $\mathcal{I}_2$;
observations in $\mathcal{I}_1$ and $X_i \in \mathcal{I}_2$ may be
freely used during the splitting process, while
$Y_i \in \mathcal{I}_2$ are used to determine terminal node values. In
this case, equality in \eqref{honesty-assumption} is required to hold
for $i \in \mathcal{I}_2$.
Finally, a third method satisfy honesty requires the existence of
auxiliary data $\{ W_i \}$. During the splitting stage (``model
fitting''), splits are computed \emph{as if} the response variable is
$W_i$; for example,~\eqref{square-split} is used with $\mu_L$ and
$\mu_R$ being the average of the $\{ W_i \}$. Once the tree is fully
grown, predictions (``model inference\footnote{The terminology
  `inference' is used to mean computing the predictions of an existing model,
  which is unrelated to the typical usage of `inference' in
  econometrics. The former terminology is standard in applied
  settings, used when describing a data pipeline: see the documentation
of \cite{tensorflow, sklearn}.}'') are made using $Y_i$'s as usual.
The practice of using such \emph{surrogate targets} is especially
popular in time-series prediction, where different horizons are used
in fitting and inference steps (c.f., \cite{quaedvlieg2019multi}).

In the present paper, for simplicity of notation, we shall assume that the first
scheme is used to satisfy honesty---namely, splitting decisions are based
on the feature vectors $X_i$ only. Our results extend to all three schemes.\footnote{
As in \cite{crf}, constants appearing in our bounds may change in scheme two.}

Our next assumption will ensure that each one of the $p$ axes
is chosen as the splitting coordinate with a probability bounded from below.
\begin{assumption}[Randomized Cyclic Splits]
  When computing the optimal split, the algorithm flips a a
  probability $\delta$ coin that is entirely independent of everything
  else.  The first time the coin lands heads, the first coordinate is
  chosen as the splitting coordinate; the second time, the second
  coordinate is chosen, and so on, such that on the $J$-th time the
  coin lands heads, the $(p \text{ mod } J) + 1$-th coordinate is
  chosen\footnote{We adopt the convention that
    $mJ \text{ mod } J = 0$, hence notation for adding $1$ to
    $(p \text{ mod } J)$.}.  After the random splitting coordinate is chosen,
the splitting index may still be chosen based on the observations.
\end{assumption}
This is a modification of the random splitting assumption in \cite{crf},
in which each of the $p$ axes has a probability $\delta$ of being chosen at each
split. This could be directly implemented by flipping a $p\delta$
and selecting one of the $p$ coordinates uniformly at random to
split when the coin lands heads.  Another method, studied in
\cite{athey2019} and implemented by popular libraries such as \cite{xgboost},
uses randomizes the \emph{number} of available splitting axes in each round.
Specifically, a Poisson random variable $Q$ with intensity proportional
to $\sqrt{p}$ is first realized ($Q$ is realized independently
from round to round). Afterwards, $\min(Q, p)$ many axes are uniformly
selected as potential candidates\footnote{Splitting for that
node ceases if $Q = 0$.}
for splitting in that round. Clearly,
this also leads to a lower bound on the probability of being chosen
for each coordinate $1 \leq j \leq p$.

Each of the two methods above involves two separate rounds of randomization:
first, a random variable encoding the decision to split randomly is made (i.e.,
the coin or the random variable $Q$), and second, the splitting axis
is then determined. Intuitively, our cyclic splitting assumption above forgoes the
second randomization step: in doing so, the variance of the number of times
that any coordinate is chosen is reduced. Importantly, the variance will depend only
on $\delta$ and not on $p$, which we will exploit in our proofs.

\begin{assumption}[The Splitting Algorithm is $(\alpha, k)$-Regular]
  There exists some $\alpha \in (0, 1/2)$ such that
whenever a split occurs in a node with $m$ sample points, the two
hyper-rectangles contains at least $\alpha m$ many points
each. Moreover, splitting ceases at a node only when the node contains less than $2k-1$ points
for some $k$.
\end{assumption}
This key assumption is carried over from \cite{crf} and contains two
requirements.  The first requires that no split may produce a
halfspace containing too few observations, i.e., that both halfspaces
are large when measured the by count of observations.  As shown in
\cite{Wager2015AdaptiveCO}, this implies that with
exponentially small complementary probability, the splitting axis shrinks by a
factor between $\alpha$ and $1-\alpha$, so that both halfspaces also
large in Euclidean volume (with high probability).

The second half of the assumption places an upper bound on the
number of observations in terminal nodes. Trees grown under this assumption
will necessarily be deeper as the sample size $n$
(and thus, the subsample size $s = n^{\beta}$)
increases. In particular, the predictions at leaf nodes---averages of
observations $Y_i$--- will be averages of a bounded number of
terms. An important consequence is that the variance of the tree estimator
(at any test point $x$) is bounded below (c.f., the distributional
assumption on $\Var(Y \mid X = x)$).

\begin{assumption}[Predetermined Splits]
The candidate splits considered at each node do not depend on the data $\{ Z_i \}$,
so that they are fixed ahead of time. Furthermore, the number of candidate splits
at each node is finite, and every candidate split shrinks the
the length of its splitting axis by at most a factor $\alpha$.
\end{assumption}
The predetermined splitting assumption is specific to our paper.  That
candidate splits considered at each node are data-independent is
``almost'' without loss of generality since the feature space $\X$ is
fixed. For example, implementations of random forests typically use
32-bit floating point numbers as their splitting index, so that the
assumption is automatically satisfied.  Furthermore, we note that
the assumption allows candidates splits to depend on the node, so that
the splitting process is still data-driven insofar as the sequence of splits
leading up to node depends on the observations.

One interpretation of this assumption is that it aids in data compression.
For example, suppose that all features are continuous and $\X = [0,1]^p$,
and that all candidate splits have the form $(j, k/2^{m})$ for some integers
$1 \leq j \leq p$ and $0 \leq k < 2^m$, where $m \geq 1$ is a fixed integer.
If a splitting rule such as \eqref{square-split} is used, then the optimal
split depends only on the values $\{ \lfloor 2^m X_{ij} \rfloor : 1 \leq i \leq n, 1 \leq j \leq p \}$
instead of $\{ X_{ij} \}$. Each member of the former set is an integer in
$\{ 0, \dots, 2^m-1 \}$ and thus could be represented using $m$ bits.
In particular, for a grid resolution of $2^8 = 256$---a fine grid even in moderately
large dimensions---each coordinate of the feature vectors $X_i$ may be stored in a single byte.
Since modern CPUs and graphics processors store floating point numbers in four or eight bytes,
this allows for a substantial reduction, allowing computation power to scale to to larger
datasets. The process of encoding  in this way is known as \emph{quantizing}, which is an option
supported by many popular software packages. In this way, though the predetermined split assumption
may seem at first glance restrictive, it aligns our model more closely with practice.

\begin{assumption}[Distributional Assumptions on the DGP of $(X,Y)$]
The features $X_i$ are supported on the unit cube $\X = [0,1]^p$ with a density that is
bounded away from zero and infinity. Furthermore, the functions
$x \mapsto \E(Y \mid X = x)$, $x \mapsto \E(Y^2 \mid X = x)$,
and $x \mapsto \E(Y^3\mid X=x)$ are uniformly Lipschitz
continuous. Finally, the conditional variance $\Var(Y \mid X = x)$ is bounded
away from zero, i.e., $\inf_{x \in \X} \Var(Y \mid X = x) > 0$.
\end{assumption}
The continuity and variance bound assumptions are standard. Note that a consequence
of continuity and compactness of the hypercube
is that the conditional moments up to order three are bounded.
Our results will not explicitly depend on knowledge of the density of $X$: however,
the density will affect the implicit constants that we carry throughout our proofs
(c.f., Lemma 3.2 and Theorem 3.3 in \cite{crf}).

\section{Gaussianity of Multivariate $U$-Statistics}
\subsection{Test Points and Notational Conventions}
We begin investigation of the random forest estimator in this section. As discussed in the model
introduction, the random forest estimator $\RF(x)$ at a test point $x$ is a $U$-statistic
where the kernel is tree estimator $T(x)$ marginalized over external randomizations.
This paper studies the \emph{multivariate} distribution of $\RF$, specifically the correlation
structure between $\RF(x)$ and $\RF(\bar x)$ at distinct points $x$ and $\bar x \in \X$.
Towards that end, we shall fix a collection of $q$ test points
\begin{equation}
x_1, \dots, x_q \in \X
\end{equation}
throughout the remainder of the paper. As these points will remain fixed, for notational brevity
\emph{we will omit their explicit dependence} when writing estimators. Therefore, $\RF(Z_1, \dots, Z_n)$
stands for the \emph{$q$}-dimensional estimator that is the random forest evaluated
at $x_1, \dots, x_q$, given observations $\{ Z_i : 1 \leq i \leq n \}$. As a consequence of notation,
most of our equations are to be understood in $\R^q$, with equality and and arithmetic acting
coordinate-wise. Finally, when there is no confusion, subscripts (typically $k$ or $\ell$)
typically denote a specific coordinate, i.e., the estimate at the $k$-th or $\ell$-th test point;
a notable exception is $T_1$, which refers to a Hajek projection that we now describe.

\subsection{Hajek Projections}
We start by reviewing properties of the H\"oeffding Decomposition of $U$-statistics,
also known as Hajek projections; see \cite{vdv} for a textbook treatment of the univariate case.
Let $f(Z_1, \dots, Z_m) \in \R^q$ be a generic $q$-dimensional statistic based on $m$ observations.
The \emph{Hajek projection} of $f$ is defined to be
\begin{equation}
\mathring f(Z_1, \dots, Z_m) = \sum_{i=1}^m \E[f(Z_1, \dots, Z_m) \mid Z_i]
- (m-1) \E f(Z_1, \dots, Z_m).
\end{equation}
That is, it is the coordinate-wise projection of $f$ to the linear space spanned
by functions of the form $\{ g(Z_i) : 1 \leq i \leq m \}$. In particular,
when $f$ is symmetric in its arguments and $Z_1, \dots, Z_m$ is an IID sequence,
we have
\begin{equation}
\label{eq:generichajek}
\mathring f(Z_1, \dots, Z_m) = \sum_{i=1}^m f_1(Z_i) - (m-1)\E f,
\end{equation}
where is the function such that $f_1(Z_1) = \E(f \mid Z_1)$, i.e., $f_1(z) = \E(f \mid Z_1 = z)$.

In our setting, applying the Hajek projection to the \emph{centered} statistic $\RF - \mu$, where $\mu$
is the expectation of $\RF$, yields
\begin{equation}
\mathring\RF(Z_1, \dots, Z_n) - \mu = \sum_{i=1}^n \E(\RF - \mu\mid Z_i)
=
\frac{1}{\binom{n}{s}}
\sum_{i=1}^n \E
\biggl[
\sum_{i_1, \dots, i_s} \E_{\xi}T(\xi; Z_{i_1}, \dots, Z_{i_s}) - \mu \mid Z_i
 \biggr],
\end{equation}
where $i_1, \dots, i_s$ run through the $\binom{n}{s}$ size-$s$ subsets of
$\{ 1, \dots, n \}$ as usual. (Recall that $\RF$, $\mu$, and $T$ are all vectors in $\R^q$.)
Since the samples $Z_1, \dots, Z_n$ are independent,
$\E(\E_{\xi}T(\xi;Z_{i_1}, \dots, Z_{i_s}) \mid Z_i) = \mu$
whenever $i \notin \{ i_1, \dots, i_s \}$. As $\{ i_1, \dots, i_s \}$
runs over the size-$s$ subsets of $\{ 1, \dots, n \}$,
there are exactly $\binom{n-1}{s-1}$ many which contain~$i$.
For each of of these subsets,
\begin{equation}
\E(\E_{\xi} T(\xi; Z_{i_1}, \dots, Z_{i_s}) - \mu \mid Z_i) \eqqcolon
T_1(Z_i) - \mu,
\end{equation}
where $T_1(z) = \E_{\xi, Z_2, \dots, Z_s} T(\xi; z, Z_2, \dots, Z_s)$.
Therefore,
\begin{equation}
\label{eq:usum}
\mathring\RF - \mu =
\frac{1}{\binom{n}{s}} \sum_{i=1}^n \biggl( {n-1 \atop s-1} \biggr) (T_1(Z_i) - \mu)
=
\frac{s}{n} \sum_{i=1}^n (T_1(Z_i) - \mu).
\end{equation}

The sequence of observations $Z_1, \dots, Z_n$ is assumed to IID, and this property
is preserved for the sequence $\{ T_1(Z_i) : 1 \leq i \leq n \}$ of projections.
It is easily verified that $\E(\mathring \RF) = \mu$, and the point of the previous
equation is that it expresses the centered statistic $(\mathring \RF - \mu)$ as an average
of centered IID terms, scaled by $s$. This will be our main entry point
in establishing asymptotic joint normality.

\subsection{Asymptotic Gaussianity via Hajek Projections}
The standard technique in deriving the asymptotic distribution
of a $U$-statistic is to establish a lower bound on the
variance of its Hajek projection; this is the approach taken by \cite{crf}
and we follow the approach here.
Let $V$ be the variance of $\mathring\RF$; using \eqref{eq:usum}, we have
\begin{equation}
\label{eq:Vvariance}
V = \Var\biggl[
\frac{s}{n} \sum_{i=1}^n (T_1(Z_i) - \mu)
 \biggr] =
\frac{s^2}{n} \Var(T_1(Z_1)) =
\frac{s}{n} \Var\biggl[
\sum_{i=1}^s T_1(Z_i)
 \biggr]
=
\frac{s}{n} \Var \mathring T \in \R^{q \times q},
\end{equation}
where $\mathring T$ is the Hajek projection of the statistic $T$ as in
\eqref{eq:generichajek}, where $T = \E_{\xi}T(\xi, Z_1, \dots, Z_s) \in \R^q$.

Since the third moment $\E(Y^3 \mid X = x)$ is assumed to be bounded,
conditions for the Lindberg Central Limit Theorem \cite{billingsley2008probability}
easily follows, and applying the triangular CLT, we have the familiar fact
\begin{equation}
V^{-1/2}(\mathring\RF - \mu) \Dto N(0, I),
\end{equation}
where $0$ is the zero vector in $\R^q$ and $I$ the $q \times q$ identity matrix.
\begin{remark}
  The condition that the third moment is bounded is not necessary. The Lindberg
conditions were directly verified in \cite{crf} (c.f., Theorem~8) and their results---without assuming
a bounded third moment---apply to our setting as well. More recently, triangular array CLTs specific to
$U$-statistics were developed in \cite{diciccio2020clt}, and their conditions are satisfied
for our case as well.
\end{remark}

The asymptotic normality of the random forest estimator $\RF$
can be related to the asymptotic normality of $\mathring \RF$ via
\begin{equation}
V^{-1/2}(\RF - \mu) =
V^{-1/2}(\RF - \mathring \RF) +
V^{-1/2}(\mathring \RF - \mu).
\end{equation}
Since the second summand on the RHS is asymptotically normal,
by Slutsky's Theorem, $V^{-1/2}(RF-\mu)$ is asymptotically normal
once we establish the
convergence $V^{-1/2}(\RF - \mathring \RF) \Pto 0$. The strategy is
to show that $e = V^{-1/2}(\RF - \mathring \RF)$ converges in squared mean.
We may develop its squared norm via
\begin{equation}
\label{eq:error squared}
  \begin{split}
\E (e^{\intercal} e) &=
\E(\RF - \mathring \RF)^{\intercal} V^{-1} (\RF - \mathring \RF)
= \E \tr V^{-1}(\RF - \mathring \RF)(\RF - \mathring \RF)^{\intercal} \\
&= \tr V^{-1} \E(\RF - \mathring \RF)(\RF - \mathring \RF)^{\intercal}
= \tr V^{-1/2} \Var(\RF - \mathring \RF)V^{-1/2},
  \end{split}
\end{equation}
where we used the identity $\tr(ABC) = \tr(BCA)$ for conforming matrices
$A$, $B$, and $C$. That trace on the extreme RHS goes to zero
is the natural multivariate generalization of the
familiar condition
\begin{equation}
\frac{\Var(f - \mathring f)}{\Var \mathring f} \to 0
\end{equation}
for univariate $U$-statistics \cite{vdv}.

In the univariate setting, the previous
condition is checked by considering higher order decompositions
proceeds by considering higher order decompositions of the statistic $f$;
this approach is also valid in the multivariate setting, as shown below. As we will
see, the more substantive
difficulty is that the dimension of the kernel of the $U$-statistic---namely, the subsample
size $s$ of each tree---grows with the sample size.

By Proposition~\ref{prop:hoeffding} below, we may expand $\RF - \mathring \RF$
according to a H\"oeffding decomposition taken
\emph{respect to the matrix} $V^{-1}$,
\begin{equation}
\RF - \mathring \RF = \frac{1}{\binom{n}{s}}
\biggl[
\sum_{i<j} \binom{n-2}{s-2} (T^{(2)}(Z_i, Z_j) - \mu) +
\sum_{i < j < k} \binom{n-3}{s-3} (T^{(3)}(Z_i, Z_j, Z_k) - \mu) +
\cdots
 \biggr],
\end{equation}
where $T^{(2)}$, $T^{(3)}$, etc.\ are the second and third order projections
of $T$ which obey the normal equations\footnote{
These normal equations are the main contents of Proposition~\ref{prop:hoeffding}.}
\begin{equation}
\label{eq:1}
\E[ (T^{(k)} - \mu)^{\intercal} V^{-1} (T^{(k')} - \mu)] = 0,
\qquad \text{for $k \neq k'$.}
\end{equation}
Of course, the higher order terms $T^{(k)}$, being projections of $T$,
also satisfy
\begin{equation}
\label{eq:2}
\E[(T^{(k)} - \mu)^{\intercal} V^{-1} (T^{(k)} - \mu)] \leq
\E[(T - \mu)^{\intercal} V^{-1} (T - \mu)].
\end{equation}

These two relationships, used with \eqref{eq:Vvariance} and \eqref{eq:error squared},
imply that
\begin{equation}
\label{eq:3}
\E(e^{\intercal}e) \leq \frac{s}{n} \tr(\Var \mathring T^{-1} \Var T).
\end{equation}
The remainder of this section centers around proving
that the quantity on the RHS converge to zero. For comparison,
a central result of \cite{crf} (using our notation) is the bound
on the \emph{diagonal elements} of $\Var \mathring T^{-1}$
and $\Var T$. Specifically, the authors obtain
\begin{equation}
\label{eq:4}
\frac{(\Var T)_{kk}}{(\Var \mathring T)_{kk}} \leq c (\log s)^p,
\qquad \text{for each $k = 1, \dots, q$},
\end{equation}
for some constant $c$. As we will see in the next section, the required bound on the trace
will follow by developing bounds on the \emph{off-diagonal} elements
of $\Var \mathring T$, i.e., bounds on the covariance between random forest estimates
at different test points (see discussion following Proposition~\ref{prop:trace-bound}).

\begin{proposition}[H\"oeffding Decomposition for Multivariate $U$-statistics]
  \label{prop:hoeffding}
Fix a positive definite matrix $M$.
Let $f(x_1, \dots, x_n) \in \R^p$ be a vector-valued function
that is symmetric in its arguments and let $X_1, \dots, X_n$
be a random sample such that $f(X_1, \dots, X_n)$ has finite variance. Then
there exists functions $f_1, f_2, \dots, f_n$ such that
\begin{equation}
\label{eq:5}
f(X_1, \dots, X_n) = \E(f) +  \sum_{i=1}^n f_1(X_1)
+ \sum_{i < j} f_2(X_i, X_j) + \dots + f_n(X_1, \dots, X_n)
\end{equation}
where $f_k$ is a function of $k$ arguments, such that
\begin{equation}
\label{eq:6}
\E f_k(X_1, \dots, X_k) = 0 \quad \text{and}\quad
\E [f_k(X_1, \dots, X_k)^{\intercal} M f_{\ell}(X_1, \dots, X_l)]  = 0.
\end{equation}
\end{proposition}
\begin{proof}
 (All proofs may be found in the Appendix.)
\end{proof}

\subsection{Covariance Bounds}
The aim of this section is to establish asymptotic bounds
on the off-diagonal elements of the the covariance matrix
$\Var \mathring T$. We shall show that when the tree estimator employs
splitting algorithms satisfying suitable \emph{stability conditions},
we have the asymptotic behavior
\begin{equation}
\label{eq:covariance bounds}
(\Var \mathring T)_{k, l} = o(s^{-\epsilon})
\text{ for all $1 \leq k \neq l \leq q$ and some $\epsilon > 0$}.
\end{equation}
Before proceeding, we first show that that this bound, coupled with
control on the diagonal terms~\eqref{eq:4}, suffices to establish
the trace bound in~\eqref{eq:3} vanish.
\begin{proposition}\label{prop:trace-bound}
The entries of $\Var T$ are bounded and its diagonal entries are bounded away from zero. 
Furthermore, when $\Var \mathring T$ satisfies the condition
in~\eqref{eq:covariance bounds},
\begin{equation}
\label{eq:7}
\frac{s}{n} \tr(\Var \mathring T^{-1} \Var T) \to 0.
\end{equation}
\end{proposition}
\begin{remark}
  The first part of the Proposition,
concerning the entries of $\Var T$, is a
consequence of our $(\alpha, k)$-regularity assumption
and distribution assumptions on $\{ Z_i \}$.
As discussed in the assumption section, since the number of observations
in leaf nodes are bounded above, the (pointwise) variance of the tree estimator
at $x$ is bounded above by (a constant times) $\Var(Y \mid X = x)$,
and we assumed that the latter function is bounded away from zero.
That the off-diagonal entries are bounded is a trivial consequence of
the fact that $\E(Y^2 \mid X = x)$ is Lipschitz and thus bounded.
The techniques we present
to bound $\Var \mathring T$ could also be used to bound 
$\Var T$; it is in fact true that $(\Var T)_{k, l} \to 0$ 
for $k \neq l$, though we will not pursue this further in this paper.
\end{remark}

Proposition~\ref{prop:trace-bound} establishes $\mathring T$
as the central object of study. Recall that
$T$ is the tree estimator while $\mathring T$
is its Hajek projection; in other words,
$\Var \mathring T$ is \emph{not} covariance matrix of tree estimates.
However, our result will demonstrate the the asymptotic normality $V^{-1}(\RF - \mu)$,
where $V$, the variance of the Hajek projection $\mathring\RF$, is given in
terms of $\Var \mathring T$ (c.f.,~\eqref{eq:Vvariance}). Therefore,
(a rescaled version of) $\mathring T$ is precisely the object 
needed to conduct inference on the random forest. In particular,
combining \eqref{eq:4} and \eqref{eq:covariance bounds} yields the
fact that $\Var \mathring T$---and hence the asymptotic variance of $\RF$---is diagonally
dominant (i.e., tending to a diagonal matrix in the limit).

We may always relabel indices so that the tree is grown on the observations $Z_1, \dots, Z_s$.
To establish the bound \ref{eq:7}, start with the definition
\begin{equation}
\label{eq:8}
\mathring T - \mu = \sum_{i=1}^s \E(T \mid Z_i)
\quad \text{so that}\quad
\Var \mathring T = s \Var(\E(T \mid Z_1))
\text{ due to independence}.
\end{equation}
To develop the term on the RHS, use the orthogonality
condition for conditional expectation
\begin{equation}
\label{eq:9}
\Var \E(T \mid Z_1) = \Var[\E(T \mid Z_1) - \E(T \mid X_1)] +
\Var[\E(T \mid X_1)].
\end{equation}
Since the tree algorithm is honest, the difference
$\E(T \mid Z_1) - \E(T \mid X_1)$ simplifies, so that
for each $1 \leq k \leq q$,
\begin{equation}
\label{eq:10}
\E(T_k \mid Z_1) - \E(T_k \mid X_1) = \E(I_k \mid X_1) (Y_1 - \E(Y_1 \mid I_k = 1, X_1)),
\end{equation}
where $T_k$ is the tree estimate at $x_k$, and $I_k$ the indicator for whether
$X_1$ and $x_k$ belong to the same terminal node.
Therefore, off-diagonal entry at $(k, l)$ of $\Var[\E(T \mid Z_1) - \E(T \mid X_1)]$
is equal to
\begin{equation}
\label{eq:11}
\E[\E(I_k \mid X_1) \E(I_l \mid X_1)
(Y_1 - \E(Y_1 \mid X_1, I_k = 1)
(Y_1 - \E(Y_1 \mid X_1, I_l = 1)].
\end{equation}

If we expand the terms in the integrand, every term will
have the shape
\begin{equation}
\E(I_k \mid X_1)\E(I_l \mid X_1) \cdot p(Y_1, \E(Y_1 \mid X_1, I_k=1), \E(Y_1 \mid X_1, I_l=1))
\end{equation}
for some multinomial $p$ of degree at most two. Since we have assumed that
$\E(Y^2 \mid X = x)$ and $\E(Y^2 \mid X = x)$ are continuous and hence bounded,
$\E(p \mid X_1 = x)$ is also bounded. Using the Law of Iterated Expectations
to evaluate \eqref{eq:11} then shows that it is bounded by a constant times
\begin{equation}
\E[\E(I_k \mid X_1) \E(I_{j} \mid X_1)].
\end{equation}
\begin{remark}
A direct application of the Cauchy-Schwarz inequality, using only the fact that
$\E(Y \mid X=x)$ is bounded, would yield the weaker bound
\begin{equation}
\sqrt{\E[\E(I_k \mid X_1)^2 \E(I_{j} \mid X_1)^2]}
\leq \sqrt{\E[\E(I_k \mid X_1) \E(I_{j} \mid X_1)]}
\end{equation}
up to a multiplicative constant. 
\end{remark}

Recall that $I_k$ and $I_l$ are indicator variables for whether $X_1$
belong to the same hypercube as $x_k$ and $x_l$, respectively. Therefore,
$\E(I_k \mid X_1)$ is the probability that the first observation is used
for the prediction at $x_k$, and likewise for $\E(I_l \mid X_1)$. Intuitively,
this only happens when $X_1$ is near $x_k$ (respectively, $x_l$):
since $x_k \neq x_l$, $X_1$ cannot be near to both, meaning that the
product $\E(I_k \mid X_1)\E(I_l \mid X_1)$ is small.
\begin{proposition}
\label{prop:terminal-node-probability}
For two points $x$ and $\bar x \in \X = [0,1]^p$, define
\begin{equation}
\label{eq:M}
M(x, \bar x) = \E[\E(I \mid X_1) \E(\bar I \mid X_1)],
\end{equation}
where $I$ and $\bar I$ are indicators for $X_1$ belonging to the same terminal
node as $x$ and $\bar x$, respectively. If $\delta > 1/2$ and $x \neq \bar x$,
\begin{equation}
\label{eq:12}
M(x, \bar x) = o(s^{-(1+\epsilon)}) \quad \text{for some $\epsilon > 0$}.
\end{equation}
\end{proposition}
\begin{remark}
It is instructive to consider the bound in the preceding display 
versus $M(x,x)$. It is clear from the definition that $M(x, x) \geq M(x, \bar x)$
for all $\bar x$.
In addition, $M(x, x) = \E(\E(I \mid X_1)^2) \leq \E(\E(I \mid X_1)) = \E(I)$.
By symmetry, $\E I = 1/s$ (up to constant), as the terminal node at $x$ has a bounded
number of observations. Therefore, all that the Proposition ensures is that when $x \neq \bar x$,
the quantity $M(x, \bar x)$ is smaller than the ``trivial'' bound $1/s$.
\end{remark}
This proposition shows that the contribution of
$\Var[\E(T \mid Z_1) - \E(T \mid X_1)]$ to the cross covariances
of $\Var \E(T \mid Z_1)$ is small, in particular
smaller than the required bound $(\log^p s \cdot s)^{-1}$. 
The requirement that $\delta > 1/2$, while needed
for the proof to go through, is almost certainly not needed in practice. The reason
is that our proof uses $\delta > 1/2$ to derive a \emph{uniform} bound on the quantity
\begin{equation}
\E(I_k \mid X_1) \E(I_l \mid X_1),
\end{equation}
while proposition demands a bound on its expectation. Indeed, in the
extreme case $x = 0$ and $\bar x = (1, \dots, 1)^{\intercal}$, it is
easy to see that the expectation satisfies the required bound even
when $\delta \leq 1/2$.  

Furthermore, our proof is agnostic to the exact splitting rule used by
the base tree learner and uses only ``random splits'' (c.f., Assumption 2)
in derive the required bounds. With a specific splitting rule (e.g., \eqref{square-split})
and a specific data distribution, the expectation $M(x,\bar x)$ will
be smaller than that predicted by \eqref{eq:12}.
In light of this, an alternative to our cyclic splitting assumption
is to assume the \emph{high level} condition
that the splitting algorithm and data generating process
confer the bound
\begin{equation}
M(x, \bar x) = o\biggl( \frac{1}{\log^p s \cdot s} \biggr).
\end{equation}

\subsubsection{Bounding $\Var \E(T \mid X_1)$}
We turn next to bound the off-diagonal terms in $\Var[\E(T \mid X_1)]$.
As in the statement of Proposition~\ref{prop:terminal-node-probability},
it will be convenient to slightly change notations. We fix $1 \leq k \neq l \leq q$,
and use the notation $x \mapsto x_k$, $\bar x \mapsto x_l$, with $x_1$
denoting the value of $X_1$. The goal of this section
is to establish the bound
\begin{equation}
\label{eq:covariance-bound}
\Var[\E(T \mid X_1)]_{kl} = \E(\E(T \mid X_1 = x_1) - \mu)
(\E(\bar T \mid X_1 = x_1) - \bar \mu))
= o\biggl( \frac{\log^ps}{s} \biggr).
\end{equation}
where $T$ and $\bar T$ are the tree estimates at $x$ and $\bar x$,
and $\mu$ and $\bar \mu$ are the (unconditional) expectations
of $T$ and $\bar T$. (Note that we have also changed the notation of $T$ slightly;
before $T$ was the $q$-dimensional vector of the estimate at all points; for this section,
it is the pointwise estimate at $x = x_k$.)

The quantity $\E(T \mid X_1 = x_1) - \mu = \E(T \mid X_1 = x_1) - \E(T)$
measures the degree of ``information'' that the location of a single observation
$X_1$ carries for the output of the tree at $x$. Intuitively, when~$x_1$ is near~$x$,
the effect of $x_1$ on the leaf node containing $x$ is more pronounced and we expect
$\E(T \mid X_1 = x_1) - \mu \approx \E(Y \mid X_1 = x_1) - \mu$. Conversely, when $x_1$
is far from $x$, then its effect in determining the location of the leaf node
containing $x$ diminishes, and $\E(T \mid X_1 = x_1) - \mu \approx 0$.

The key in making the above intuition precise is to keep track of the
$X_1 = x_1$ leaves the intermediate partition containing $x$,
where ``intermediate partition'' refers to the nodes created during the splitting
process that are not necessarily terminal. We will see that once $X_1$ and
$x$ are separated, its effect on the prediction decreases.

Towards this end, fix $x$ and let $\Pi$ denote the terminal node containing $x$;
$\Pi$ is a random subset of $\X$ created by axis aligned splits.
By Assumption 4, the set of potential splits does not depend on 
the sample (in particular, it does not depend on $X_1$). Moreover, splitting ceases
after no more than $s$ times, regardless of the subsample $X_1, \dots, X_s$, as each split
reduces the number of observations in its two children nodes by at least one. Therefore,
$\Pi$ takes on only finitely many possible values, and we may write
\begin{equation}
\label{eq:tree-formula}
\E(T) = \sum_{\pi} \P(\Pi = \pi) \mu_{\pi}
\quad \text{and}\quad
\E(T \mid X_1 = x_1) = \sum_{\pi} \P(\Pi = \pi \mid X_1 = x) \mu_{\pi}'
\end{equation}
where $\mu_{\pi} = \E(T \mid \Pi = \pi)$ and
$\mu_{\pi}' = \E(T \mid \Pi = \pi, X_1 = x)$.

The hyperrectangle $\Pi$ is determined by the recursive splitting procedure used
to grow the tree, and there is a natural correspondence between \eqref{eq:tree-formula}
and a certain ``expectation'' taken over the \emph{directed acyclic graph} (DAG)
defined in the following way.
Let $[0, 1]^p$ be the root of the DAG; for every potential split at $[0,1]^p$,
there is a directed edge to a new vertex, where the vertex is the 
hyperrectangle that contains $x$. If the node represented by a vertex
is one of the possible values of $\Pi$, then that vertex is a leaf in the DAG and has no outgoing edges;
other vertices carry an outgoing edge for each potential split at that node,
with each edge going to another vertex which is again a hyperrectangle containing $x$.

The previous definition determines the DAG recursively: each vertex in the DAG is a node containing $x$,
with terminal vertices corresponding to terminal nodes. To each terminal vertex $v$, we associate
the value $f(v) \coloneqq \mu_{\pi}$ as in \eqref{eq:tree-formula}.
In addition, each edge
$e = (v \to w)$ corresponds to a split at a node $v$ producing a halfspace $w$ of $v$;
associate with this edge a ``transition probability''
\begin{equation}
\label{eq:24}
p(e) \coloneqq \P(\text{$s$ is chosen at $v$} \mid \text{current node is $v$})
\eqqcolon \P(w \mid v).
\end{equation}
Given the transition probabilities, the value $f$ may be extended
to each vertex $v$ recursively via the formula
\begin{equation}
\label{eq:25}
f(v) \coloneqq \sum_{e: v \to w}  \P(w \mid v) f(w).
\end{equation}
We refer to $f$ as the continuation value at $v$, and by construction
we have
\begin{equation}
\label{eq:22}
\E(T) = f(\text{``root''}) = f([0,1]^p).
\end{equation}

Alternatively, if we had assigned the values $f'(v) = \mu'_v$
to each terminal vertex and used the transition probabilities
\begin{equation}
\label{eq:27}
p'(e) = \P(\text{$s$ is chosen at $v$} \mid \text{current node is $v$}, X_1 = x_1)
= \P'(w \mid v),
\end{equation}
then we recover $\E(T \mid X_1 = x_1) = f'([0,1]^p)$ after extending $f'$ in the same way
as $f$. In other words, bounding $\E(T \mid X_1 = x_1) - \E(T)$ requires bounding the difference
between the two types of continuation values.

We will need to assume that $p'(e) \approx p(e)$; that is, conditioning on a \emph{single}
observation will not affect the probability that a particular split is chosen (i.e., that a particular
split is optimal at its node). This is a natural assumption in that the optimal
split is computed using all the observations in a particular node, so that conditioning
on a single observation should have relatively little effect. This will depend on the specifics
on the splitting algorithm used to a construct a tree, and rules which satisfying the following
assumption to be ``stable splitting rules.''
\begin{splitting-stability}
For any node $v$, the total variation distance between the distributions
$\{ p(e) \}_{e: v \to w}$ and $\{ p'(e) \}_{e: v \to w}$ is bounded by
the volume of $v$. Specifically, there exists some $\epsilon>0$
such that for all $v$,
\begin{equation}
\label{eq:29}
\operatorname{TV}(p, p') \leq \biggl( \frac{1}{s|v|} \biggr)^{1+\epsilon}\qquad
\text{(up to a constant).}
\end{equation}
Here, $|v|$ denotes the volume of the hyperrectangle at $v$, i.e.,
\begin{equation}
\label{eq:30}
|v| = \biggl| \prod_{j=1}^p (a_j, b_j) \biggr| = \prod_{j=1}^p |b_j - a_j|.
\end{equation}
\end{splitting-stability}
\begin{remark}
Since $p$ and $p'$ are discrete probability distributions: thus, if 
$p$ and $p'$ are written as vectors of probability masses, then the total
variation distance is the $L_1$ norm between the two vectors.
\end{remark}
Since the distribution of $X$ has a density that is bounded above and below,
a simple H\"oeffding bound shows that the number of sample points in $v$ is
bounded above and below by $s|v|$, with the constants adjusted so that 
the failure probability is less than\footnote{For example,
the probability that a binomial random variable $B(n,p)$
deviates from $np$ by more than $\sqrt{n \log n}$ is less than $C/s^2$
for some constant $C$.}
$1/s^2$. Since this is smaller than the required bound $(\log^p s \cdot s)^{-1}$,
we may interpret $s|v|$ in \eqref{eq:29} to be the number of samples in $|v|$
without loss of generality.
Relatedly, \cite{Wager2015AdaptiveCO}'s Lemma 12 (see also the proof
of Lemma~2 in \cite{crf}) extends to this fact to be uniform across nodes.

The stability assumption places a restriction on procedure used to
select optimal splits: namely, if the decision is made on the basis of
$m$ points, then conditioning on any one of the points changes the
optimal split with probability bounded by $m^{-(1+\epsilon)}$. In
practice, most splitting procedures satisfy a stronger bound. A set of
sufficient conditions is given in the following proposition.
\begin{proposition}
\label{prop:split-stability-sufficient}
Assume that the optimal split at a node $v$ is chosen based on the quantities
\begin{equation}
f_1(\mu_1, \dots, \mu_Q), \dots, f_P(\mu_1, \dots, \mu_Q)
\end{equation}
for some $Q \geq 1$, where
$\mu_1, \dots, \mu_Q$ are the sample averages of the points being split
\begin{equation}
\mu_k  = \frac{1}{n_v} \sum_{i: X_i \in v} m_k(X_i)
\end{equation}
where the sum runs over points in $v$, and $n_v$ denote the number of
these points. 

Specifically, suppose optimal split is decided based on which $f_i$ achieves the largest value,
i.e., the value $\argmax_i f_i(\mu)$. If
$f_1, \dots, f_P$ are Lipschitz, and the functions $m_1, \dots, m_Q$ are such
that $m_k(X)$ is 1-subExponential, then the splitting stability assumption is satisfied.
\end{proposition}
\begin{remark}
Since $X_i$ are bounded, the requirement that $m_k(X)$ is subExponential
allows the use of \eqref{square-split} to compute the optimal split.
\end{remark}
In general, the conditions in Proposition~\ref{prop:split-stability-sufficient}
are sufficient to guarantee an exponential bound instead of a polynomial one 
as in~\eqref{eq:29}. Thus, Proposition~\ref{prop:split-stability-sufficient}
should be viewed as simply as providing a plausibility argument
that stable splitting rules are commonly encountered in practice.

The next proposition shows that bounds on splitting probabilities
automatically imply a related bound on the continuation values.
\begin{proposition}
\label{prop:value-bounds}
Suppose the splitting probabilities satisfy a generic bound $\Delta(\cdot)$ in that
\begin{equation}
\TV(p, p') \leq \frac{\Delta(s|v|)}{\log s} \quad \text{at each node $v$}.
\end{equation}
For example, $\Delta(z) = z^{-(1+\epsilon)}$. 
Then for any node $v$ containing $x$ but not $x_1$,
\begin{equation}
|f(v) - f'(v)| \leq C \Delta(s|v|)
\end{equation}
for some constant $C$ not depending on $v$.
\end{proposition}
The splitting stability assumption stipulates that $\Delta(z) = z^{-(1+\epsilon)}$,
where the factor $\epsilon$ allows us to ignore the extra logarithm. 
In that case, we may put the bounds on $\TV(p, p')$ and $|f-f'|$
together and establish required bound on $\Var \E(T \mid X_1)$.
\begin{proposition}\label{prop:2-bound}
Suppose that the splitting rule is stable as in \eqref{eq:29}
and that $\delta  > 1-\alpha$.
For $x \neq x_1$,
\begin{equation}
\label{eq:26}
|\E(T  \mid X_1 = x_1) - \E(T)| = o\biggl( \frac{1}{s^{1+\epsilon}} \biggr)
\end{equation}
for some $\epsilon > 0$. In particular, the off-diagonal
entries of $\Var \E(T \mid X_1)$ are $o(s^{-(1+\epsilon)})$
as at least one of $x$ and $\bar x$ is distinct from $x_1$.
\end{proposition}
Recall that Proposition~\ref{prop:terminal-node-probability} required
that $\delta > 1/2$. Since $\alpha < 1/2$ by definition, the requirement
that $\delta > 1-\alpha$ in Proposition~\ref{prop:2-bound} is more restrictive.
Just like Proposition~\ref{prop:terminal-node-probability},
we argue that this requirement is plausibly looser in 
applications. The reason is that it is used to give the following bound
on hyperrectangles $v$ created after $L$ splits
\begin{equation}
|v| \geq \alpha^L.
\end{equation}
The RHS appears since potential split may reduce the volume of a node by at most $\alpha$:
but only an exponentially small (i.e., $2^{-L}$) proportion of nodes are the result of
taking the smallest possible split $L$ times!
The ``average'' node has volume approximately
$(1/2)^L$, so that $\delta > 1/2$ may be more appropriate. 

\subsection{Wrapping Up}
Combining Propositions~\ref{prop:terminal-node-probability} and \ref{prop:2-bound}
with equations \eqref{eq:8} and \eqref{eq:9} yields the desired bound \eqref{eq:covariance bounds} on the off-diagonal
terms of $\Var \mathring T$ discussed at the beginning of this section.
Therefore, Proposition~\ref{prop:trace-bound} applies, and the joint
normality of the random forest estimator is established.

\section{Heuristics and Simulations}
The previous sections focused on deriving the asymptotic normality result
\begin{equation}
\label{eq:32}
V^{-1/2}(\RF - \mu) \Dto N(0,I),
\qquad \text{where } 
V = \Var \mathring \RF = \frac{s}{n} \Var \mathring T.
\end{equation}
Recall our standing convention that $\RF\in \R^q$ is the random forest estimate
at $x_1, \dots, x_q$ and $\mu$ is its expectation.
According to \eqref{final-prediction}, the target function of the random forest is actually
is $m(x) = \E(Y \mid X = x)$. The results in \cite{crf} show that
$(\RF(x) - m(x))/\sqrt{V} \Dto N(0, 1)$ pointwise for each $x \in \X$; since we have shown
that $V$ is diagonally dominant in that its off-diagonal terms vanish relative to the diagonal, 
the pointwise result carries over to our multivariate setting, and $V^{-1/2}(\RF - m) \Dto N(0, I)$,
where $m = (m(x_1), \dots, m(x_q))$.

Moreover, \cite{crf} proposes a jackknife estimator that can consistently
estimate $\sqrt{V}$ in the scalar case. Our diagonal dominance result implies that
the random forest estimates at $x$ and $\bar x \in [0, 1]^p$ are independent
in the limit $n\to\infty$,
\begin{equation}
\label{eq:33}
\begin{split}
\Var(\RF(x) + \RF(\bar x)) &= \Var(\RF(x)) + \Var(\RF(\bar x)) + 2 \Cov(\RF(x) + \RF(\bar x))
\\&\approx \Var(\RF(x)) + \Var(\RF(\bar x)),
\end{split}
\end{equation}
so that the jackknife estimator for the scalar case
may be fruitfully applied to obtain confidence bands for \emph{functionals}
of the random forest estimates (i.e., expressions involving the estimate at more than one point).

The accuracy of the approximation above above depends on the decay of
the off-diagonal terms in finite samples. In this section, we provide a
``back of the envelope'' bound for the covariance term that may be useful
for practitioners. We stress that the following calculations are (mostly)
heuristics: as we have shown above, the covariance term depends on
quantities such as $M(x, \bar x)$, which is in turn heavily dependent on the
exact mechanics of the underlying splitting algorithm. 
Since our aim is to produce a ``usable'' result, we will dispense with rigorous analysis
in the remainder of this section.

To begin, the proofs of Propositions~\ref{prop:terminal-node-probability} and~\ref{prop:2-bound}
showed that the asymptotic variance $V$ has off-diagonal
terms which are upper bounded by\footnote{This is a very crude upper bound
as we have dropped the quantity $\Delta(\alpha^{\ell} s)$
from the infinite series.}
\begin{equation}
M(x, \bar x) + \log^2s \biggl(\sum_{\ell = 0}^{\infty} p_{\ell}\biggr)
\biggl(\sum_{\ell = 0}^{\infty} \bar p_{\ell}\biggr)
=
M(x, \bar x) + \log^2s \E(L)\E(\bar L)
\end{equation}
where $p_{\ell} = \P(L \geq \ell)$ is the probability that $x$ and
$x_1$ are not separated after $\ell$ splits and likewise for
$\bar p_{\ell} = \P(\bar L \geq \ell)$. That is, $L$ is the number of splits before
which $x$ and $x_1$ belong to the same partition.
If we denote by $I$ (resp.\ $\bar I$) the indicator variable
that $X_1$ is in the terminal node of $x$ (resp.\ $\bar x$), then
the events $\{ I=1 \}$ and $\{L = \log_2 s\}$ are equal, so that
\begin{equation}
\E(I \mid X_1 = x_1) = \P(L = \log s) \leq \frac{\E L}{\log s}.
\end{equation}
Replacing the inequality with an approximation, we have
$\E L = (\log s) \E(I \mid X_1 = x_1)$. All of this shows that
the covariance term is bounded by
\begin{equation}
(\log^4s) \E(I \mid X_1 = x_1) \E(\bar I \mid X_1 = x_1) \approx
(\log^4s) M(x, \bar x).
\end{equation}
\begin{remark}
Taken loosely, this heuristic says that the random forest
estimator $\RF$, considered as a function on the domain $\X$,
is asymptotically Gaussian with covariance process $(\log^4 s) \cdot M(x, \bar x)$.
We stress that this is \emph{not} implied by our theoretical results, as there
we kept the number $q$ of test points fixed.
\end{remark}

Towards a useful heuristic, we will consider a bound on the correlation
instead of the covariance. In our notation, the result of \cite{crf} lower bounds $M(x, x)$
(and $M(\bar x, \bar x)$), while our paper provides an \emph{upper}
bound on $M(x, \bar x)$. Ignoring the logarithmic terms, we have
\begin{equation}
\biggl|\frac{\Cov(\RF(x), \RF(\bar x))}{\sqrt{\Var \RF(x) \cdot \Var \RF(\bar x)}}\biggr|
\approx
\frac{M(x, \bar x)}{\sqrt{M(x, x) M(\bar x, \bar x)}}.
\end{equation}
Recall that $M(x, \bar x) = \E[\E(I \mid X_1) \E(\bar I \mid X_1)]$,
which decays as $\bar x$ moves away from $x$. Using the previous expression
(note that $M(x,x) \approx M(\bar x, \bar x)$ due to symmetry between $x$ and $\bar x$),
we can bound the correlation from purely geometric considerations. Since the integrand
\begin{equation}
\E(I \mid X_1) \E(\bar I \mid X_1)
\end{equation}
decays as $X_1$ moves away from $x$ (and $\bar x$), we may imagine that
its integral
\begin{equation}
\label{eq:13}
M(x, x) = \int_{x_1} \E(I \mid X_1=x_1)^2 dx_1
\end{equation}
has the largest contributions for points $x_1$ near $x$,
say those points in a $L_{\infty}$-box of side lengths $d$ with volume $d^p$, i.e.,
$\{ y \in [0,1]^p : \| x - y \| \leq d/2\}$. 
If we accept this, then the contributions for the integral $M(x, \bar x)$ would come from points
that are with $d/2$ of both $x$ \emph{and} $\bar x$, and to a first degree approximation,
the volume of these points $\{ y \in [0,1]^p : \| x - y \|_{\infty} \leq d/2, \| \bar x - y \| \leq d/2 \}$
is
\begin{equation}
(d - z_1)  \dots (d - z_p) \approx d^p - (z_1 + \dots + z_p)d^{p-1}, \quad \text{where $z_i = |x_j - \bar x_j|$},
\end{equation}
where the approximation is accurate if $|z_i| \ll 1$. Dividing through by $d^p$, the
proportion of the volume of the latter set is $1 - \frac{1}{d}\| x - \bar x \|_1$, which leads to the heuristic
\begin{equation}
\biggl|\frac{\Cov(\RF(x), \RF(\bar x))}{\sqrt{\Var \RF(x) \cdot \Var \RF(\bar x)}}\biggr|
\approx
1 - c \| x - \bar x \|_1, \qquad \text{for some constant $c$}.
\end{equation}
The RHS has the correct scaling when $x = \bar x$, i.e., the correlation equals
one when $\| x - \bar x \|_1 = 0$. To maintain correct scaling at the other extreme
with $\| x - \bar x \|_1 = p$, we should take $c = 1/p$, so that 
\begin{equation}
\biggl|\frac{\Cov(\RF(x), \RF(\bar x))}{\sqrt{\Var \RF(x) \cdot \Var \RF(\bar x)}}\biggr|
=
1 - \frac{1}{p}\sum_{i=1}^p |x_i - \bar x_i|.
\end{equation}

Of course, this heuristic must be incorrect in that it does not depend on $s$;
our theoretical results show that even for non-diametrically opposed points,
the correlation drops to zero as $s \to \infty$. Therefore,
another recommendation is to use
\begin{equation}
\label{linear-heuristic}
\biggl|\frac{\Cov(\RF(x), \RF(\bar x))}{\sqrt{\Var \RF(x) \cdot \Var \RF(\bar x)}}\biggr|
=
\min \biggl(
1 - \frac{s^{\epsilon}}{p}\sum_{i=1}^p |x_i - \bar x_i|, \; 0
\biggr),
\end{equation}
for some $\epsilon > 0$, where dependence $s^{\epsilon}$  comes from considering the decay
of $M(x, \bar x)$ as $\bar x$ moves away from $x$ (c.f.\ the proof Proposition~\ref{prop:terminal-node-probability}).

\subsection{Simulations}
In this section, we discuss the results of simulations calculating the correlation structure.
For our experiment, we set $p = 2$, so that
the covariates $X$ are distributed on the unit square. The distribution of $X$ is chosen to be
``four-modal''
\begin{equation}
\label{eq:16}
X \sim
\begin{cases}
\bar N(\mu_1, I_2)  & \text{with probability $1/4$} \\
\bar N(\mu_2, I_2) & \text{with probability $1/4$}\\
\bar N(\mu_3, I_2) & \text{with probability $1/4$}\\
\bar N(\mu_4, I_2) & \text{with probability $1/4$}
\end{cases} \quad \text{where}\quad
\begin{cases}
\mu_1 = (0.3, 0.3)^{\intercal} \\
\mu_2 = (0.3, 0.7)^{\intercal} \\
\mu_3 = (0.7, 0.3)^{\intercal} \\
\mu_4 = (0.7, 0.7)^{\intercal} \\
\end{cases}
\end{equation}
and $\bar N$ denotes a truncated multivariate Gaussian distribution on the unit square.\footnote{
That is, $\bar N(\mu, \Sigma)$ denotes the conditional distribution of $x \sim N(\mu, \Sigma)$
on the event $x \in [0, 1]^2$.} Thus, $X$ has a bounded density on the unit square, and has four peaks
at $\mu_1, \dots, \mu_4$. The distribution of $Y$ conditional on $X = (x_1, x_2)$ is
\begin{equation}
\label{eq:18}
Y \sim \frac{x_1 + x_2}{2} + \frac{1}{5} N(0, 1).
\end{equation}
The random splitting probability is $\delta = 1/2$, and the regularity parameters
are $\alpha = 0.01$ and $k = 1$, so that the tree is grown to the fullest extent (i.e., terminal
nodes may contain a single observation), with each terminal node lying on the $101 \times 101$
grid of the unit square. For each sample size $n$, five thousand trees are grown, and the
estimates are aggregated to compute the correlation.

Figure~\ref{fig:1} plots the correlation of estimates at $x$ and $\bar x$
as a function of the $L_1$ norm $\|x - \bar x\|_1$. The calculation is performed by
first fixing $x$, then calculating the sample correlation (across five thousand
trees) as $\bar x$ ranges over each cell: the correlation is associated with
the $L_1$ norm $\|x - \bar x\|_1$. This process is then repeated by varying the reference point $x$,
and the correlation at $\| x - \bar x \|_1$ is the average of the correlations observed.
\begin{figure}[t]
  \centering
\includegraphics[width=\textwidth]{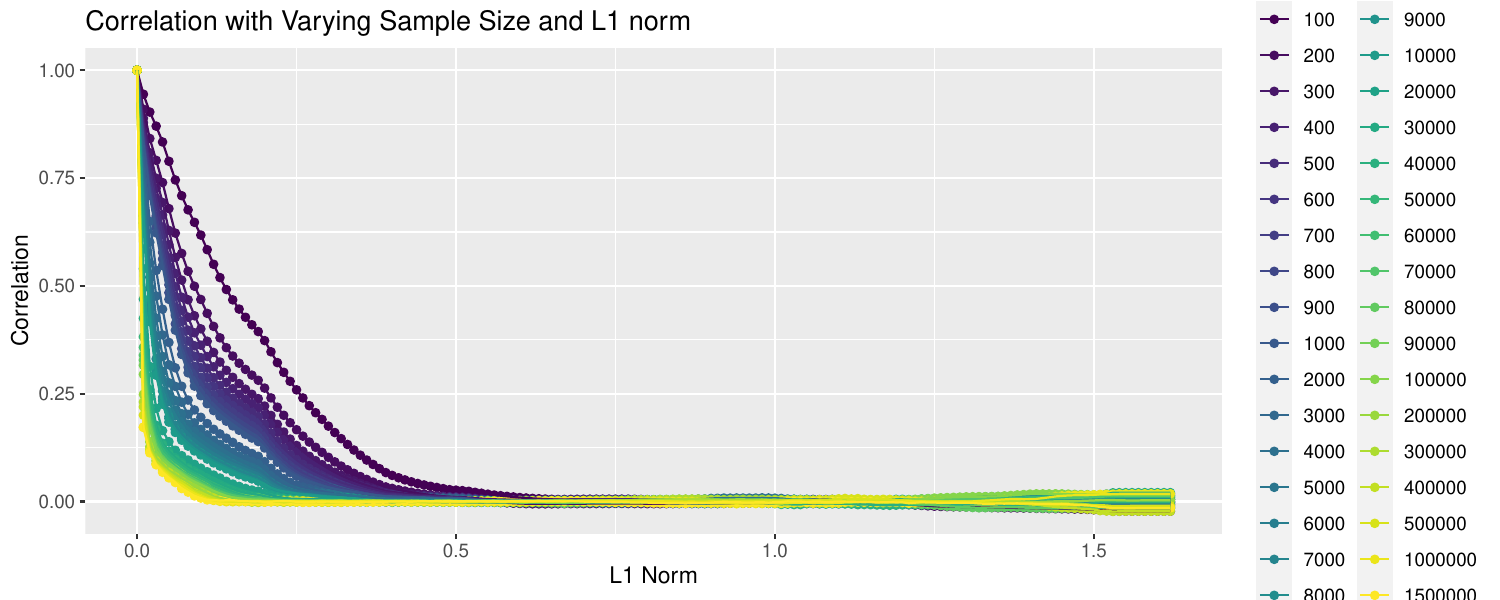}
  \caption{Correlation as a function of sample size and $L_1$ norm.}
  \label{fig:1}
\end{figure}
The figure demonstrates that the linear heuristic \eqref{linear-heuristic} given in the previous
section is conservative: it is evident that correlation decreases super-linearly as $x$ and $\bar x$
become separated.

Figure~\ref{fig:2} plots the correlation on a logarithmic scale, which shows that
that correlation decay is exponential in a neighborhood of unity. In other words, simulations
suggest that the correct heuristic is of the shape
\begin{equation}
\label{eq:28}
\biggl|\frac{\Cov(\RF(x), \RF(\bar x))}{\sqrt{\Var \RF(x) \cdot \Var \RF(\bar x)}}\biggr| \approx
e^{-\lambda\| x - \bar x \|_1}
\end{equation}
for a suitable $\lambda$.
\begin{figure}[t]
  \centering
\includegraphics[width=\textwidth]{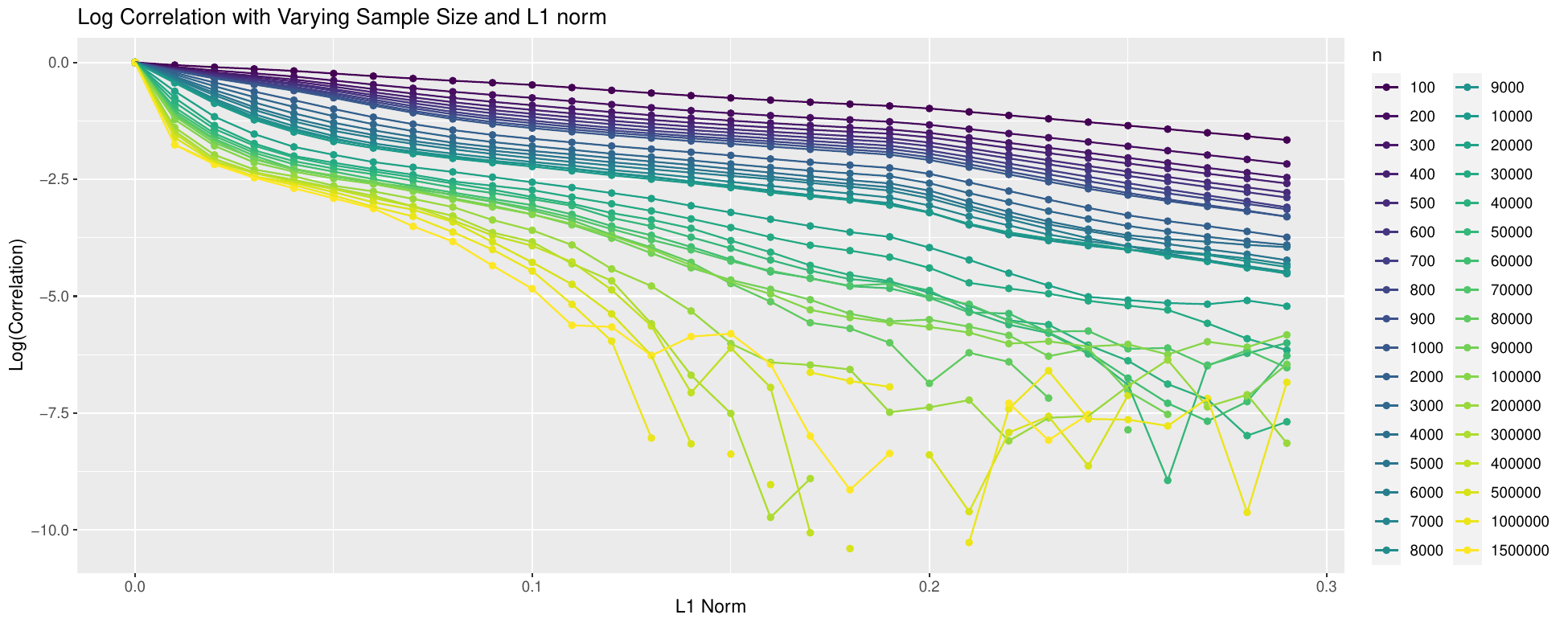}
  \caption{Logarithm of correlation as a function of sample size and $L_1$ norm.}
  \label{fig:2}
\end{figure}

\section{Conclusion}
Random forests and tree-based methods form an important part of of an
applied data analysis toolkit. In this paper, we studying the
covariance between random forest estimates at several points. We
develop a novel construction of a directed acyclic graph that keeps
track of the splitting probabilities when knowledge of one point is
known (Propositions~\ref{prop:value-bounds} and
\ref{prop:2-bound}). As part of the proof, we establish
stability properties of a class of splitting rules (see Proposition
\ref{prop:split-stability-sufficient}). We also identify (Proposition~\ref{prop:terminal-node-probability})
$M(x, \bar x)$, which (roughly) captures the likelihood of two points
belonging to the same terminal node, as a key quantity in controlling the off-diagonal term
of the covariance matrix of the multivariate random forest. 

In this way, this paper provides the a theoretical basis for performing inferences
on functionals of target function (e.g., a heterogeneous treatment effect) when
the functional is based on values of the target function at multiple points in the feature space.
Specifically, we show that the covariance vanishes in the limit relative to the variance,
and provide heuristics on the size of the correlation in finite samples. 

We close with discussing a couple avenues for future research.
The first is extending our framework to cover
categorical or discrete-valued
features. Here, new assumptions would be required in order to maintain
the guarantee that node sizes are ``not too small.'' Second, our
bounds---after potential improvements---on the covariance matrix
of the random forest may be used with the recent results of
\cite{chernozhukov2017,chen2018} in order to provide finite sample
Gaussian approximations. This would provide a sounder theoretical underpinning of
for our heuristics, and increase the usefulness of this paper for practitioners.

\appendix
\section{Proofs}
\begin{proof}[Proof of Proposition~\ref{prop:hoeffding}]
For random vectors in $\R^q$, define the inner product
\begin{equation}
\label{eq:14}
\langle X, Y \rangle \coloneqq \E(X^{\intercal} M Y).
\end{equation}
For each subset $A \subseteq \{ 1, \dots, n \}$, let $H_A$
be the set of square-integrable random vectors of the form
$g(X_i : i \in A)$, where $g$ is a function of $|A|$ arguments,
satisfying the condition that
\begin{equation}
\E(g(X_i : i \in A) \mid \{ X_i : i \in B \}) = 0
\end{equation}
for all subsets $B \subsetneq A$. It is easy to see that 
collection $H_{A}$ are pairwise orthogonal as $A$ ranges over subsets $\{ 1, \dots, n \}$.
By induction on $r=|A|$, the direct sum $\bigoplus_{B \subseteq A} H_B$ is equal
to the set of all statistics which are functions of $\{ X_i : i \in A \}$.
In particular, $\bigoplus_A H_A$ is the set of all statistics based on $\{ X_1, \dots, X_n \}$.
When the variables $\{ X_1, \dots, X_n \}$ are IID, then $H_A$ depends only on $|A|$
in that there exist collections of functions $H_0, H_1, \dots, H_n$, where $H_k$ is a collection
of $k$-ary functions, such that
\begin{equation}
H_A = \{ g(X_i : i \in A) : g \in H_{|A|} \}.
\end{equation}
The proof is complete by letting $f_k$ be the projection of $f$
onto $H_k$ according to the inner product given in \eqref{eq:14}.
\end{proof}

\begin{proof}[Proof of Proposition~\ref{prop:trace-bound}]
We will prove the slightly more general statement that if $A_n$
and $B_n$ are two sequences of square matrices with bounded entries
such that
\begin{equation}
\label{eq:15}
B_{ii} > \delta \text{ for some $\delta$ for all $n$}\quad\text{and}\quad
A_{ii} \geq \frac{B_{ii}}{\log n}
\end{equation}
and $A_{ij} = o(1/\log n)$, then $\tr(A^{-1} B)\to 0$. To prove this, start with the determinant
formula $\det A = \sum_{\pi} (-1)^{\operatorname{sgn} \pi} \prod_{i=1}^q a_{i \pi_i}$,
where the sum runs over permutations $\pi$
of $\{ 1, \dots, n \}$ and $\operatorname{sgn} \pi$ is the sign of the permutation.
Since the off-diagonal entries of $A_{ij}$ are assumed to vanish relative to $A_{ii}$,
we have $|\det A| \sim \prod A_{ii}$, 
where the notation $a \sim b$ stands for $c |b| \leq |a| \leq c' |b|$ for constants
$c$ and $c'$ not depending on $n$.
Next, recall Cramer's Rule
\begin{equation}
(A^{-1})_{ii} = \det A_{-i}/\det A,
\end{equation}
where $A_{-i}$ is the matrix $A$ with its $i$-th row and $i$-th column removed.
A similar argument shows that $|\det A_{-i}| \sim \prod_{j \neq i} A_{jj}$, whence
\begin{equation}
\label{eq:16}
(A^{-1})_{ii} \sim \frac{1}{A_{ii}}.
\end{equation}
In particular, the $i$-th diagonal entry of the matrix $A^{-1} B$ is given by
\begin{equation}
\label{eq:17}
(A^{-1}B)_{ii} = (A^{-1})_{ii} B_{ii} + \sum_{j \neq i} (A^{-1})_{ij} B_{ji}
\sim  \frac{B_{ii}}{A_{ii}} \leq \log n,
\end{equation}
where the final relation is due to the fact that $(A^{-1})_{ij}$ is itself
a polynomial in the entries of $A$ (viz., the cofactor matrix of $A$) divided by the determinant.
Therefore, the trace of $A^{-1}B$ is on the order of $\log n$, since
the dimension $q \times q$ of each matrix is fixed. Using the subsample size $s = n^{\beta}$,
so that $s/n = n^{-(1-\beta)}$ completes the proof.
\end{proof}

\begin{proof}[Proof of Proposition~\ref{prop:terminal-node-probability}]
Recall that the splitting algorithm has a probability $\delta$ chance
of splitting on the $j$-th axis. Since each terminal node contains a
constant number of points, the number of terminal nodes is equal (up to constant)
to the subsample size $s$. Therefore, the number of splits 
required to reach a terminal node is bounded (by a constant) by $\log_2 s/K = \log s/K$,
where $K=2k-1$ is the the maximum size of a leaf.

Since $x \neq \bar x$, we have
\begin{equation}
0 < \| x - \bar x \|_{\infty} \leq \| x - x_1 \|_{\infty} + \| \bar x - x_1 \|_{\infty}
\end{equation}
for all $x_1 \in \X$. In particular, given any $x_1$ there exists some $j \in \{ 1, \dots, p \}$
and a constant $\beta$ for which either $|x_j - x_{1j}| > \beta$
or $|\bar x_j - x_{1j}| > \beta$. Without loss of generality, we may assume that the former case holds.
Certainly, a necessary condition for $X_1 = x_1$ to belong to the same leaf node
as $x$ (i.e., a necessary condition for $\{ I = 1 \}$) is for the length
of the first axis of that leaf node to be larger than $\beta$. 

Let $c_j(x)$ denote the number of splits in coordinate $j$ along the sequence
of splits leading to the terminal node containing $x$. By our randomization assumption,
each split has at least an independent chance $\delta$ of being chosen,
and since we cycle through each coordinate (c.f., Assumption 2),
\begin{equation}
\label{eq:18}
c_j(x) \succeq \frac{1}{p} B\biggl(\log\frac{s}{K}, \delta\biggr)
\qquad \text{where $\succeq$ stands for stochastic dominance}.
\end{equation}
Per Assumption 4, that each split along the $j$-th axis decreases its
length by a factor of at least $(1-\alpha)$. Since splitting
begins in the unit hypercube, 
\begin{equation}
\label{eq:19}
(1-\alpha)^{c_1(x)} \geq \beta \implies
c_1(x) \leq \frac{\log \beta}{\log (1-\alpha)} \eqqcolon \rho.
\end{equation}
Since $\{ I = 1 \}$
requires that the length of the first axis to exceed $\beta$ (a constant), this proves
\begin{equation}
\E(I \mid X_1 = x_1) \leq 
\P\biggl[ B\biggl( \log\frac{s}{K}, \delta \biggr) \leq p \rho \biggr].
\end{equation}
Since $p\rho$ is a constant, we may conclude
\begin{equation}
\label{eq:20}
\P\biggl[ B\biggl( \log\frac{s}{K}, \delta \biggr) \leq p \rho \biggr]
\leq (1 - \delta + o(1))^{\log s/K} = \biggl( \frac{1}{s} \biggr)^{\log \frac{1}{1-\delta + o(1)}}.
\end{equation}

Finally, recall base of the logarithm is two since the tree is binary. Therefore,
if we choose $\delta > 1/2$, the exponent exceeds $1$ and the proof is complete.
\end{proof}

\begin{proof}[Proof of Proposition~\ref{prop:split-stability-sufficient}]
The easiest case is the splitting decision in the root node $[0, 1]^p$, so we start
here. We prove the result by introducing a coupling between the split decisions 
with and without conditioning on $X_1 = x_1$
\begin{equation}
\label{eq:34}
\begin{split}
  S &= \argmax_i f_i
\biggl( \frac{1}{s} \sum_{i=1}^s m_1(X_i),
\dots, \frac{1}{s} \sum_{i=1}^s m_Q(X_i) \biggr) \eqqcolon f_i\\
S' &= \argmax_i f_i \biggl( \frac{1}{s}\biggl[m_1(x_1) + \sum_{i=2}^s m_1(X_i) \biggr],
\dots,
\frac{1}{s}\biggl[m_Q(x_1) + \sum_{i=2}^s m_Q(X_i) \biggr]
\biggr) \eqqcolon f_i'.
\end{split}
\end{equation}
Here, $S$ is the split made on the sample $X_1, \dots, X_s$
and $S'$ is the split made on the sample conditional on $X_1 = x_1$. Note that 
we may assume without loss of generality
that the splits are not randomly chosen, since on that event the splitting
probabilities are trivially equal.
Clearly, a necessary condition for $S \neq S'$ is
the existence of a pair $1 \leq i \neq j \leq P$ for which
\begin{equation}
\label{eq:35}
f_i > f_j \quad \text{but} \quad f_j' > f_i'.
\end{equation}
Since $f$ is Lipschitz and its arguments are subExponential by assumption,
the quantities $f_i$ and $f_j$ concentrate around their respective limits $f_i(\E m_1, \dots, \E m_Q)$
and $f_j(\E m_1, \dots, \E m_Q)$;
hence, whenever $f_i(\E m_1, \dots, \E m_Q) > f_j(\E m_1, \dots, \E m_{Q})$
we will have
\begin{equation}
\label{eq:36}
f_i - f_j > \frac{1}{s}
\quad \text{with probability at least $1 - O(e^{-cs})$ for some constant $c$.}
\end{equation}
However, the difference of the arguments of $f_i$ in (\ref{eq:34})
differ by at most $1/s$, the difference due to the $|m_1(x_1)-m_1(X_1)|/s$. By Lipschitz
continuity, a change of $1/s$ in the arguments changes the function values by a proportional
amount, so $f_j' > f_i$ is impossible when $f_i - f_j > 1/s$. It follows that 
$(f_i > f_j, f_j' > f_i')$ occurs with probability at most $O(e^{-cs})$,
and we finish by noting taking a union over the $\binom{P}{2}$ pairs $(i,j)$. 
This result result for $[0, 1]^p$ will be referred to as the base case.

Note that the above actually proves something stronger, namely that for every
split $\tau$,
\begin{equation}
\label{eq:37}
\P(S \neq s \mid S' = \tau) < e^{-cs} \quad \text{and} \quad
\P(S' \neq s \mid S = \tau) < e^{-cs}.
\end{equation}
It follows that see that for any $\tau$, the total variation
distance between $(X_2, \dots, X_s \mid S = \tau)$
and $(X_2, \dots, X_s \mid S' = \tau)$ is at most $e^{-cs}$.
To see this, note that $S$ and $S'$ are functions of $X_i$ only,
so that the densities of these two distributions are
\begin{equation}
\label{eq:38}
p(x) = \1(S(x) = \tau) \frac{p(x)}{\P(S = \tau)}
\quad \text{and}\quad
p'(x) = \1(S'(x) = \tau) \frac{p(x)}{\P(S' = \tau)}
\end{equation}
respectively. We may assume without loss of generality that $\P(S' = \tau) \geq \P(S = \tau)$
so that the total variation is
\begin{equation}
\label{eq:39}
\begin{split}
\int |p(x) - p'(x)|
&= \int_{S = \tau} p(x) - p'(x) + \int_{S \neq \tau, S' = \tau}
p'(x) \\
&= 1 - \frac{\P(S' = \tau, S = \tau)}{\P(S' = \tau)}
+ \frac{\P(S' = \tau, S \neq \tau)}{\P(S'=\tau)}  
= 2\P(S \neq \tau \mid S' = \tau) < e^{-cs}.
\end{split}
\end{equation}
The upshot is that when considering the splitting probability in the next node,
we can ignore the difference in the distribution of $X_2, \dots, X_s$
when conditioning on $S = \tau$ versus conditioning on $S' = \tau$
and pay a cost $O(e^{-cs})$.

Now consider bounding the difference of the splitting probabilities at the
next split
\begin{equation}
\label{eq:40}
\P(S_2 = s \mid S_1 = \tau) - \P(S_2 = s \mid S_1 = \tau,  X_1 = x_1).
\end{equation}
Again, the strategy is to find a coupling $(S_2, S_2')$ such that 
\begin{equation}
\label{eq:41}
S_2 \sim (S_2 \mid S_1 = \tau) \quad \text{and}\quad 
S_2' \sim (S_2 \mid S_1= \tau, X_1 = x_1)
\end{equation}
with $\TV(S_2, S_2') \leq e^{- s|v|}$, $|v|$ being the hyper-rectangle corresponding
to one of the halfspaces produced by $\tau$. Since the distribution of $X_1, \dots, X_s$
conditional on $S_1 = \tau$ differs from its unconditional distribution
by an amount $e^{-cs}$ in the total variation distance, we could use the following
coupling
\begin{equation}
\label{eq:42}
\begin{split}
  S_2 &= \argmax f_i\biggl( \frac{1}{n_v} \sum_{X_i \in v} m_1(X_i), \dots,
\frac{1}{s|v|} \sum_{X_i \in v} m_Q(X_i) \biggr)\\
  S_2' &= \argmax f_i\biggl( \frac{1}{n_v} \sum_{X_i'\in v} m_1(X_i'), \dots,
\frac{1}{s|v|} \sum_{X_i \in v} m_Q(X_i') \biggr)
\end{split}
\end{equation}
where $X_i$ follows the distribution of $(X_1, \dots, X_s)$ conditional
on $S_1 = \tau$ and $X_i'$ follows the distribution conditional on $S_1 = \tau$
and $X_1 = x_1$.  By conditioning on $S_1 = \tau$ instead of $S_1' = \tau$,
we increase the total variation by an amount $e^{-cs}$ via the triangle inequality.

Now the rest of the proof is the same as in the base case, noting that with high probability,
the number $n_v$ of points in $v$ is equal to $s|v|$ up to an multiplicative constant $(1-\eta)$
with probability $e^{-s \eta^2}$. The previous bounds
are applied recursively at each depth $l$ of the DAG. 
At depth $l$, we incur an ``approximation cost'' from the total
variation distance bounded by $e^{-|v|s}$. Since
$s|v| \geq \alpha^l$, it follows that $l \leq O(\log s|v|)$,
whence the cumulative cost at depth $l$
is $O(\log s|v| \cdot e^{-c |v| s})$.
Putting everything together, we have proven that
\begin{equation}
\label{eq:43}
\TV(p, p') \leq O(\log (|v| s) \cdot e^{-c |v| s}) \leq
 o\biggl( \frac{1}{s|v|} \biggr)^{1+\epsilon}
\quad \text{for some $\epsilon > 0$}.\qedhere
\end{equation}
\end{proof}
\begin{proof}[Proof of Proposition~\ref{prop:value-bounds}]
The claim is trivially true (by choosing an appropriate constant)
if $v$ is a terminal node. Thus, fix a non-terminal node $v$
such that $x_1 \notin v$ and let
\begin{equation}
\mathbf{X} = \mathbf{X}_v = \{ X_i : X_i \in v\}
\end{equation}
denote the set of points landing in $v$, so that
$k \coloneqq |\mathbf{X}| \in \{ 1, \dots, n-1 \}$.

Recall that $f=f(v)$ and $f=f'(v)$ are the respective expectations of the tree estimator
at $x$ when the sequence of splits is such that $v$ is the current subset of $\X$
containing $x$, with $f'$ being calculated conditional on $X_1 = x_1$. It follows that
$f$ and $f'$ are functions of the distribution of its ``input vector'' $\mathbf{X}$.
In a slight abuse of notation, let $\Pi$ and $\Pi'$ denote sequence of splits
distributed according to the probabilities $p$ and $p'$. 
We will show that, for each $k \in \{ 1, \dots, n-1 \}$, the total variation distance
of
\begin{equation}
(\mathbf{X} \mid |\mathbf{X}| = k, \Pi = v)
\quad \text{and} \quad
(\mathbf{X} \mid |\mathbf{X}| = k, \Pi' = v, X_1 = x_1)
\label{twodist}
\end{equation}
is bounded by $(\log s) \cdot \Delta(s|v|)$.  This will suffice to bound $|f-f'|$
by the variational definition of total variation distance
\begin{equation}
\label{eq:tvvariation}
\operatorname{TV}(p, p') = \sup_{|g| \leq 1} | \E_{A \sim p} g(A) - \E_{A \sim p'} g(A)|.
\end{equation}
Since $x_1 \notin v$,
$X_1$ is not an element of $\mathbf{X}$, so that $\mathbf{X}$ and $X_1$ are independent.
Since the split $\Pi'$ is distributed according to splitting probabilities when $X_1 = x_1$,
we have, 
\begin{equation}\label{eq:removex}
\dist(\mathbf{X} \mid |\mathbf{X}| = k, \Pi' = v, X_1 = x_1) =
\dist(\mathbf{X} \mid |\mathbf{X}| = k, \Pi' = v).
\end{equation}
The depth of  $v$ is at most $\log_2 s$, so that
$\P(\Pi = \Pi') \geq 1 - (\log_2 s) \Delta(s|v|)$
by applying the splitting stability assumption $\log_2 s$
many times using a union bound.
By \eqref{eq:removex} the total variation distance
of the distributions in \eqref{twodist} differs by $\P(\Pi \neq \Pi')$,
and the result follows.
\end{proof}
\begin{proof}[Poor of Proposition~\ref{prop:2-bound}]
The idea is to recursively expand the formulas $\E T$ and $\E(T \mid X_1)$
in terms of the directed acyclic graph. We start with
\begin{equation}
  \begin{split}
|\E T - \E(T \mid X_1 = x_1)| &=
\biggl| \sum_v p(e) f(e) - \sum_v p'(e) f'(e) \biggr| \\
&\leq \sum_v |p(e) - p'(e)| f'(e)
+ \biggl| \sum_v p'(e) (f(e) - f'(e)) \biggr|,
  \end{split}
\end{equation}
where the sum runs over nodes $v$ after the first split,
i.e., $[0,1]^p \to v$. The second summand may be split
into two, one over $x_1 \in v$ and the other $x_1 \notin v$.
If we assume splitting stability
held with function $\Delta$, then Proposition~\ref{prop:value-bounds}
allows us to bound the second term, so that
\begin{equation}
  \begin{split}
|\E T - \E(T \mid X_1 = x_1)|
&\leq \sum_v |p(e) - p'(e)| f'(e)
+ \biggl| \sum_v p'(e) (f(e) - f'(e)) \biggr|\\
&\leq \Delta(\alpha s) + (\log s) \Delta(\alpha s) +
\sum_{x_1 \in v} p'(e) |f(e) - f'(e)|,
  \end{split}
\end{equation}
where used the fact that $|v| \geq \alpha$.
Now, each of the terms $|f(e) - f'(e)|$ may be bounded
by $\Delta(\alpha^2 s) + \log(s) \Delta(\alpha^2 s) + \sum_{x_1 \in w}(\cdots)$.
Continuing in this way, we have
\begin{equation}
|\E T - \E(T \mid X_1 = x_1)| \leq
\log(s) (\Delta(s) + p_1\Delta(\alpha s) + p_2 \Delta(\alpha^2 s) + \dots)
= \log s \sum_{\ell=0}^{\infty} p_{\ell} \Delta(\alpha^{\ell}s),
\end{equation}
where $p_{\ell}$ is the probability that $x$ and $x_{\ell}$
belong to the same node after $\ell$ splits. In other words,
if we let $L$ be the number of splits after which $x$ and $x_1$
are separated, then
\begin{equation}
p_{\ell} = \P(L \geq \ell).
\end{equation}

Since $x \neq x_1$, we may assume without loss of generality that 
$\| x - x_1 \| > \beta$ for some fixed $\beta$ (c.f. the proof of Proposition~\ref{prop:terminal-node-probability}).
In particular,
\begin{equation}
p_{\ell} \leq (1-\delta + o(1))^{\ell}
\end{equation}
for sufficiently large $\ell$.
Moreover, $\Delta(\alpha^l s) = \frac{1}{s^{1+\epsilon}}(\frac{1}{\alpha^{1+\epsilon}})^{\ell}$,
whence $\delta > 1-\alpha^{1+\delta}$ is enough to ensure that
infinite series is less than $\frac{1}{s^{1+\epsilon}}$.
In particular, as $s \to 0$, we may take $\epsilon \to 0$, so that the restriction
is satisfied (after suitable constants) by $\delta > 1-\alpha$. This completes
the proof.
\end{proof}
\singlespacing
\bibliographystyle{unsrt}
\addcontentsline{toc}{section}{References}
\bibliography{references}

\begin{thebibliography}{10}

\bibitem{breiman01}
Leo Breiman.
\newblock Random forests.
\newblock {\em Machine Learning}, 2001.

\bibitem{lightgbm}
Guolin Ke, Qi~Meng, Thomas Finley, Taifeng Wang, Wei Chen, Weidong Ma, Qiwei
  Ye, and Tie-Yan Liu.
\newblock Lightgbm: A highly efficient gradient boosting decision tree.
\newblock In {\em Proceedings of the 31st International Conference on Neural
  Information Processing Systems}, NIPS'17, page 3149–3157, Red Hook, NY,
  USA, 2017. Curran Associates Inc.

\bibitem{xgboost}
Tianqi Chen and Carlos Guestrin.
\newblock {XGBoost}: A scalable tree boosting system.
\newblock In {\em Proceedings of the 22nd ACM SIGKDD International Conference
  on Knowledge Discovery and Data Mining}, KDD '16, pages 785--794, New York,
  NY, USA, 2016. ACM.

\bibitem{grfgithub}
Susan Athey, Julie Tibshirani, and Stefan Wager.
\newblock The grf algorithm.
\newblock \url{https://github.com/grf-labs/grf/}, 2017.

\bibitem{catboost}
Liudmila Prokhorenkova, Gleb Gusev, Aleksandr Vorobev, Anna~Veronika Dorogush,
  and Andrey Gulin.
\newblock Catboost: Unbiased boosting with categorical features.
\newblock In {\em Proceedings of the 32nd International Conference on Neural
  Information Processing Systems}, NIPS'18, page 6639–6649, Red Hook, NY,
  USA, 2018. Curran Associates Inc.

\bibitem{gregorutti2017correlation}
Baptiste Gregorutti, Bertrand Michel, and Philippe Saint-Pierre.
\newblock Correlation and variable importance in random forests.
\newblock {\em Statistics and Computing}, 27(3):659--678, 2017.

\bibitem{strobl2008conditional}
Carolin Strobl, Anne-Laure Boulesteix, Thomas Kneib, Thomas Augustin, and Achim
  Zeileis.
\newblock Conditional variable importance for random forests.
\newblock {\em BMC bioinformatics}, 9(1):307, 2008.

\bibitem{genuer2010variable}
Robin Genuer, Jean-Michel Poggi, and Christine Tuleau-Malot.
\newblock Variable selection using random forests.
\newblock {\em Pattern recognition letters}, 31(14):2225--2236, 2010.

\bibitem{rubin}
Donald~B. Rubin.
\newblock Estimating causal effects of treatments in randomized and
  nonrandomized studies.
\newblock {\em Journal of Educational Psychology}, 66(5):688--701, 1974.

\bibitem{imbens_rubin_2015}
Guido~W. Imbens and Donald~B. Rubin.
\newblock {\em Causal Inference for Statistics, Social, and Biomedical
  Sciences: An Introduction}.
\newblock Cambridge University Press, 2015.

\bibitem{propensity_score}
Keisuke Hirano, Guido~W. Imbens, and Geert Ridder.
\newblock Efficient estimation of average treatment effects using the estimated
  propensity score.
\newblock {\em Econometrica}, 71(4):1161--1189, 2003.

\bibitem{crf}
Stefan Wager and Susan Athey.
\newblock Estimation and inference of heterogeneous treatment effects using
  random forests.
\newblock {\em Journal of the American Statistical Association},
  113(523):1228--1242, 2018.

\bibitem{athey2019}
Susan Athey, Julie Tibshirani, and Stefan Wager.
\newblock Generalized random forests.
\newblock {\em Ann. Statist.}, 47(2):1148--1178, 04 2019.

\bibitem{Wager2015AdaptiveCO}
Stefan Wager and Guenther Walther.
\newblock Adaptive concentration of regression trees, with application to
  random forests.
\newblock {\em arXiv: Statistics Theory}, 2015.

\bibitem{arsov2019stability}
Nino Arsov, Martin Pavlovski, and Ljupco Kocarev.
\newblock Stability of decision trees and logistic regression, 2019.

\bibitem{chernozhukov2017}
Victor Chernozhukov, Denis Chetverikov, and Kengo Kato.
\newblock Central limit theorems and bootstrap in high dimensions.
\newblock {\em Ann. Probab.}, 45(4):2309--2352, 07 2017.

\bibitem{chen2018}
Xiaohui Chen.
\newblock Gaussian and bootstrap approximations for high-dimensional
  u-statistics and their applications.
\newblock {\em Ann. Statist.}, 46(2):642--678, 04 2018.

\bibitem{hastie2009elements}
Trevor Hastie, Robert Tibshirani, and Jerome Friedman.
\newblock {\em The elements of statistical learning: data mining, inference,
  and prediction}.
\newblock Springer Science \& Business Media, 2009.

\bibitem{tensorflow}
Google and other contributors.
\newblock {TensorFlow}: Large-scale machine learning on heterogeneous systems,
  2015.
\newblock Software available from tensorflow.org.

\bibitem{sklearn}
F.~Pedregosa, G.~Varoquaux, A.~Gramfort, V.~Michel, B.~Thirion, O.~Grisel,
  M.~Blondel, P.~Prettenhofer, R.~Weiss, V.~Dubourg, J.~Vanderplas, A.~Passos,
  D.~Cournapeau, M.~Brucher, M.~Perrot, and E.~Duchesnay.
\newblock Scikit-learn: Machine learning in {P}ython.
\newblock {\em Journal of Machine Learning Research}, 12:2825--2830, 2011.

\bibitem{quaedvlieg2019multi}
Rogier Quaedvlieg.
\newblock Multi-horizon forecast comparison.
\newblock {\em Journal of Business \& Economic Statistics}, pages 1--14, 2019.

\bibitem{vdv}
A.~W. van~der Vaart.
\newblock {\em Asymptotic Statistics}.
\newblock Cambridge Series in Statistical and Probabilistic Mathematics.
  Cambridge University Press, 1998.

\bibitem{billingsley2008probability}
Patrick Billingsley.
\newblock {\em Probability and measure}.
\newblock John Wiley \& Sons, 2008.

\bibitem{diciccio2020clt}
Cyrus DiCiccio and Joseph~P Romano.
\newblock Clt for u-statistics with growing dimension.
\newblock Technical report, Stanford University, 2020.

\end{thebibliography}
\end{document}